%% file: main.tex
\documentclass[12pt,hidelinks]{mitthesis}
\usepackage{amssymb}
\usepackage{lgrind}
\usepackage{cmap}
\usepackage[T1]{fontenc}
\usepackage{graphicx}
\usepackage{epsfig}
\usepackage{amsmath}
\usepackage{amsfonts}
\usepackage{times}
\usepackage{color}
\usepackage[numbers]{natbib}
\usepackage{multicol}
\usepackage[bookmarks=true]{hyperref}
\usepackage{url}
\usepackage{fullpage}
\usepackage[noend,ruled,vlined,linesnumbered]{algorithm2e}
\usepackage[font={it}]{caption}
\usepackage{subcaption}
\usepackage{amsthm}
\usepackage{graphicx,wrapfig}
\usepackage{etoolbox}
\usepackage{tocloft}
\usepackage{array}
\usepackage{makecell}
\usepackage{mathtools}
\def\all{all}
\def\files{all}
\ifx\files\all  \else   \fi

\newcommand{\abs}[1]{\left| #1\right|}
\providecommand{\norm}[1]{\left\lVert#1\right\rVert}

\newcommand{\eps}{\varepsilon}
\newcommand{\R}{\mathcal{R}}
\newcommand{\br}[1]{\left\{#1\right\}}

\newcommand{\REAL}{\ensuremath{\mathbb{R}}}

\newcommand{\loss}{\mathrm{loss}}
\newcommand{\hlf}{\mathrm{hlf}}
\newcommand{\size}{\mathrm{size}}
\newcommand{\prob}{\mathrm{prob}}
\newcommand{\proj}{\mathrm{\pi}}
\newcommand{\lip}{\mathrm{lip}}
\newcommand{\timee}{\mathrm{time}}
\newcommand{\closest}{\mathrm{closest}}
\DeclareMathOperator*{\arginf}{arg\,inf}
\DeclareMathOperator*{\argsup}{arg\,sup}
\newtheorem{theorem}{Theorem}
\newtheorem{lemma}[theorem]{Lemma}
\newtheorem{observation}[theorem]{Observation}

\newcommand{\Q}{\mathcal{Q}}
\newcommand{\of}[1]{\left(#1\right)}

\newtheorem{definition}[theorem]{Definition}
\newtheorem{corollary}[theorem]{Corollary}

\newcommand{\dist}{D}
\newcommand{\range}{\mathrm{range}}
\newcommand{\ranges}{\mathrm{ranges}}
\newcommand{\sgn}{\mathrm{sgn}}

\newcommand{\cost}{\mathrm{cost}}

\newcommand{\coreset}{\textsc{Coreset}}
\newcommand{\bcoreset}{\textsc{Coreset-Framework}}
\newcommand{\abapprox}{$(\alpha,\beta)$-approximation}
\newcommand{\bicriteria}{\textsc{Bi-Criteria-Approximation}}
\newcommand{\constantapprox}{\textsc{Centroid-Set}}
\newcommand{\sensitivityoftranslatedlines}{\textsc{Sensitivity-Bound}}
\newcommand{\sensitivityforweightedcenters}{\textsc{Weighted-Centers-Sensitivity}}
\newcommand{\streamingcoreset}{\textsc{Streaming-Coreset}}

\renewcommand{\paragraph}[1]{\medskip\noindent\textbf{{#1} }}

\DeclareMathOperator*{\argmin}{arg\,min}

\DeclarePairedDelimiter{\ceil}{\lceil}{\rceil}

\begin{document}
\include{cover}
% Some departments (e.g. 5) require an additional signature page.  See
% signature.tex for more information and uncomment the following line if
% applicable.
\include{signature}
\pagestyle{plain}
\include{contents}

\cleardoublepage
% Uncomment the next line if you do NOT want a page number on your
% abstract and acknowledgments pages.
% \pagestyle{empty}
%\setcounter{savepage}{\thepage}
\begin{abstractpage}
\addcontentsline{toc}{chapter}{Abstract}

The input to the \emph{$k$-median for lines} problem is a set $L$ of $n$ lines in $\REAL^d$, and the goal is to compute
a set of $k$ centers (points) in $\REAL^d$ that minimizes the sum of squared distances over every line in $L$ and its nearest center. This is a straightforward generalization of the $k$-median problem where the input is a set of $n$ points instead of lines.

We suggest the first PTAS that computes a $(1+\eps)$-approximation to this problem in time $O(n \log n)$ for any constant approximation error $\eps \in (0, 1)$, and constant integers $k, d \geq 1$. This is by proving that there is always a weighted subset (called coreset) of $dk^{O(k)}\log (n)/\eps^2$ lines in $L$ that approximates the sum of squared distances from $L$ to \emph{any} given set of $k$ points. 

Using traditional merge-and-reduce technique, this coreset implies results for a streaming set (possibly infinite) of lines to $M$ machines in one pass (e.g. cloud) using memory, update time and communication that is near-logarithmic in $n$, as well as deletion of any line but using linear space. These results generalized for other distance functions such as $k$-median (sum of distances) or ignoring farthest $m$ lines from the given centers to handle outliers.

Experimental results on 10 machines on Amazon EC2 cloud show that the algorithm performs well in practice.
Open source code for all the algorithms and experiments is also provided.\\

This thesis is an extension of the following accepted paper: 
"$k$-Means Clustering of Lines for Big Data", 
by Yair Marom \& Dan Feldman, 
Proceedings of \textbf{NeurIPS 2019 conference}, to appear on December 2019.
\end{abstractpage}

\newpage
\addcontentsline{toc}{chapter}{List of Tables}
\listoftables
\newpage
\addcontentsline{toc}{chapter}{List of Figures}
\listoffigures
\clearpage
\setcounter{page}{1}
\pagenumbering{arabic}

\chapter{Introduction}
\section{Background}

Clustering is the task of partitioning the input set to subsets, where items in the same subset (cluster)  are similar to each other, compared to items in other clusters. There are many different clustering techniques, but arguably the most common in both industry and academy is the $k$-mean problem, where the input is a set $P$ of $n$ points in $\REAL^d$, and the goal is to compute a set $C$ of $k$ centers (points) in $\REAL^d$, that minimizes the sum of squared distances over each point $p \in P$ to its nearest center in $C$, i.e.
$$ 
C \in \argmin_{C' \subseteq \REAL^d, \abs{C'}=k} \sum_{p \in P} \min_{c' \in C'} \norm{p-c'}^2.
$$
A very common heuristics to solve this problem is the Lloyd's algorithm \cite{advancarefulseed,theeflloyd}, that is similar to the EM-Algorithm that is described in Section~\ref{section - experimental results}.

We consider a natural generalization of this $k$-mean problem, where the input set $P$ of $n$ points is replaced by a set $L$ of $n$ lines in $\REAL^d$; See Fig.~\ref{figure - problem statement}. Here, the distance from a line to a center $c$ is the closest Euclidean distance to $c$ over all the points on the line. Since we only assume the ``weak" triangle inequality between points, our solution can easily be generalized to sum of distances to the power of any constant $z\geq1$ as explained e.g. in~\cite{charikar,oneshot} and Section~\ref{section - problem statement}.

\textbf{Motivation} for solving the $k$-line median problem arises in many different fields, when there is some missing entry in all or some of the input vectors, or incomplete information such as a missing sensor. For example, a common problem in Computer Vision is to compute the position of a point or $k$ points in the world, based on their projections on a set of $n$ 2D-images, which turn into $n$ lines via the pinhole camera model; See Fig~\ref{figure - drone}, and~\cite{liu2001using, williams2005analytical} for surveys. 

In Data Science and matrix approximation theory, every missing entry turns a point (database's record) into a line by considering all the possible values for the missing entry. E.g., $n$ points on the plane from $k$-median clusters would turn into $n$ horizontal/vertical lines that intersect "around" the $k$-median centers. The resulting problem (under similar Maximum-Likelihood arguments) is then $k$-median for $n$ lines \cite{ren2015does, shen2016online}. One can consider also an applications to semi-supervised learning - $k$-mean for mixed points and lines. This problem arises when lines are unlabeled points (last axis is a label) and we want to add a label to the farthest lines from the points.

\begin{figure} [!ht]
	\centering
	\includegraphics[width=10cm, height=6cm]{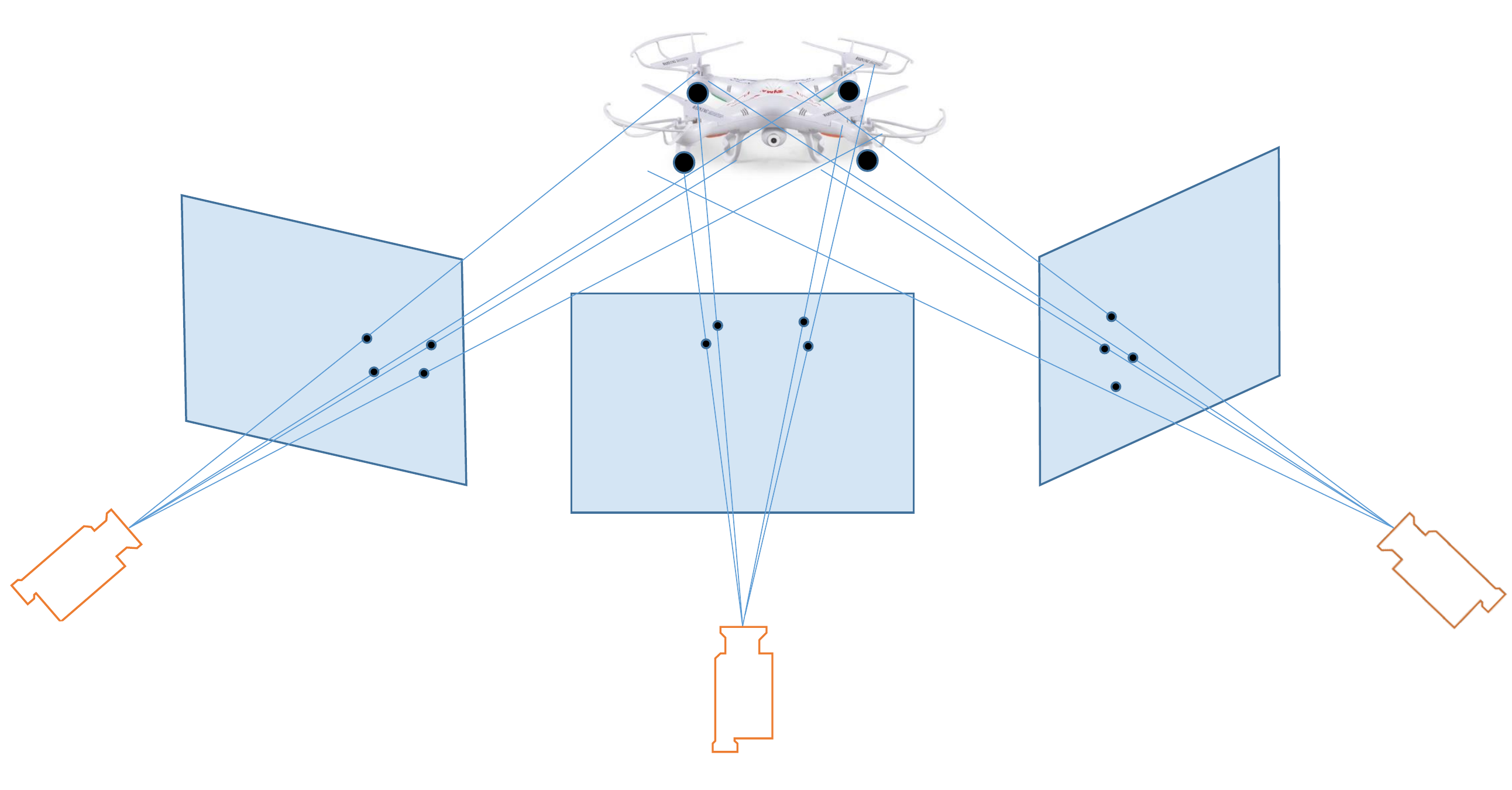}
		\caption{\textbf{Application of k-line median for computer vision.} Given a drone (or any other rigid body) that is captured by $n$ cameras - our goal is
to locate the 3 dimensional position of the drone in space by  identifying $k=4$ fixed known markers/features on the drone. Each point on each image corresponds to a line that passes through it and the pin-hole of the camera. Without noise, all the lines intersect at the same point (marker in $\REAL^3$). Otherwise their 1-median is a natural approximation.}
		\label{figure - drone}
\end{figure}

\section{Related Work}
The $k$-median problem and its variance has been researched in numerous papers over the recent decades, especially in the machine learning community~, see \cite{krause2010discriminative,nadler2006diffusion,xu2005maximum,bach2004learning} and references therein. There are also many results regarding projective clustering, when the $k$ centers are replaced by lines or $j$-dimensional subspaces instead of points.

However, significantly less results are known for the case of clustering subspaces, or even lines. A possible reason might be to the fact that the triangle inequality or its weaker version holds for a pair of points but not for lines, even in the planar case: two parallel lines can have an arbitrarily large distance from each other, but still intersect with a third line simultaneously.
Gao, Langebreg and Schulman~\cite{clusteringinhigh} used Helly's theorem \cite{hellystheorem} (intersection of convex sets) to introduce the "$k$-center problem" for lines, that aims to cover a collection of lines by the smalest $k$ balls in $\mathbb{R}^3$.

\begin{figure} %[!ht]
	\centering
	\includegraphics[width=9cm, height=5cm]{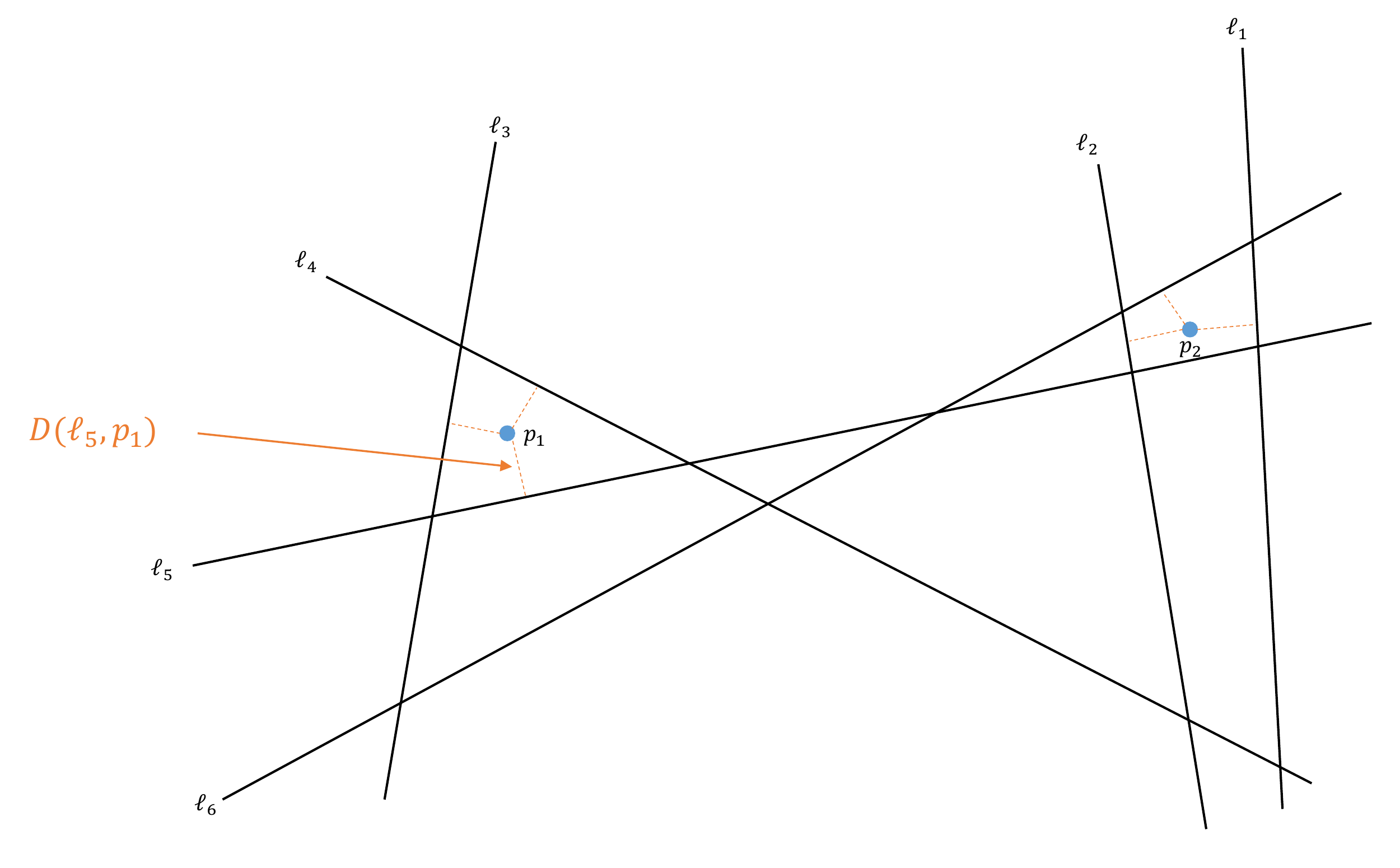}
		\caption{\textbf{Problem Statement demonstration for the planar case.} The input is a set $L=\br{\ell_1, \ldots, \ell_6}$ of $n=6$ lines in $\REAL^2$, and our goal is to find the $k=2$ points (centers) $p_1, p_2 \in \REAL^2$ that minimize the sum of Euclidean distances from each center to its nearest line in $L$. See Section~\ref{subsection - preliminaries} definition of the distance function $D$.}
		\label{figure - problem statement}
\end{figure}

Langebreg and Schulman~\cite{analysincompletedata} addressed the $1$-convex sets center problem that aims to compute a ball that intersects a given set of $n$ convex sets such as lines and $\Delta$-affine subspaces. This type of non-clustering problems $(k=1)$ is easier since it admits a convex optimization problem instead of a clustering non-convex problem.

While this thesis suggests provable solutions for clustering lines, many heuristics were provided over the years.
For example, Ommer et al.~\cite{OmmerMalik} used the Hough Transform heuristics to suggest candidates solutions in $\REAL^d$ together with other techniques such as~\cite{LeibeLeonardisSchiele} for object clustering in the transformed space. See  references therein for many other related heuristics.

Unlike the case of numerous theoretical papers that study the $k$-mean problem for points, we did not find any provable solution for the case of $k$-line median problem or even an efficient PTAS (polynomial approximation scheme). There are very few related results that we give here. A solution for the special case of $d=2$ and sum of distances was considered in~\cite{perets2011clustering}. Lee and Schulman \cite{lee2013clustering} studied the generalization of the $k$-center problem (maximum over the $n$ distances, instead of their sum), for the case where the input is a set of $n$ affine subspaces, each of dimension $\Delta$. In this case, the size of the coreset is exponential in $d$, which was proved to be unavoidable even for a coreset of $1$-center (single point) for this type of covering problems. The corresponding covering problems can then be solved using traditional computational geometry techniques.

%Our result generalizes these case for the general case of sum or sum of squared distances and $\Delta=1$ (lines). We hope that this is the first step towards such a solution for clustering an input of $n$ dimensional subspaces.

Table \ref{table - related work} summarizes the above results.
%In this paper we are the first to address the $k$-mean for lines in $\REAL^d$, for every $k,d \geq 1$.

\begin{table}[h!]
\centering
\begin{tabular}{|c|c|c|c|}
 \hline
 \textbf{Problem} & \textbf{Running Time} & \textbf{Approx. Factor} & \textbf{Papaer} \\ [1ex]
  \hline
 $1$-line center in $\REAL^d$& $nd(1/\eps)^{O(1)}$ & $1+\eps$ & \cite{analysincompletedata} \\ [1ex]
  \hline
 $2$-line centers in $\REAL^2$& $O\of{n \log\of{1/\eps}\of{d +  \log n}}$ & $2+\eps$ & \cite{clusteringinhigh} \\ [1ex]
  \hline
 $3$-line centers in $\REAL^2$& $O\of{nd \log(1/\eps)+\frac{n \log^2(n) \log(1/\eps)}{\eps}}$ & $2+\eps$ & \cite{clusteringinhigh} \\ [1ex]
  \hline
 $1$-center for convex-sets in $\REAL^d$ & $O\of{n^{\Delta+1}d (1/\eps)^{O(1)}}$ & $1+\eps$ & \cite{analysincompletedata} \\ [1ex]
  \hline
 $k$-line median in $\REAL^3$ & iterative, unbounded & unbounded & \cite{OmmerMalik} \\ [1ex]
  \hline
 $k$-line median in $\REAL^2$ & $n\of{\frac{\log n}{\eps}}^{O(k)}$ & $1+\eps$ & \cite{perets2011clustering} \\ [1ex]
 \hline
  $1$-$\Delta$-flats center in $\REAL^d$ & $O\of{\frac{nd\Delta}{\eps^2}\log\frac{\Delta}{\eps}}$ & $1+\eps$ & \cite{lee2013clustering} \\ [1ex]
 \hline
$2$-$\Delta$-flats centers in $\REAL^d$ & $O\of{dn^2\log n}$ & $O\of{\Delta^{1/4}}$ & \cite{lee2013clustering} \\ [1ex]
 \hline
$3$-$\Delta$-flats centers in $\REAL^d$ & $2^{O(\Delta(1+1/\eps^2))}nd$ & $1+\eps$ & \cite{lee2013clustering} \\ [1ex]
 \hline
$k$-$\Delta$-flats centers in $\REAL^d$ & $2^{O(\Delta k \log k(1+1/\eps^2))}nd$ & $1+\eps$ & \cite{lee2013clustering} \\ [1ex]
 \hline
$k$-line median in $\REAL^d$ & $O\of{d^2 n\log (n) k \log k} + ndk^{O(k)}$ & $1+\eps$ & \textbf{Our} \\ [1ex]
 \hline
\end{tabular}
\caption{\textbf{Summary of related results for $k$ centers of $n$ points in $\REAL^d$.} 
The dimension of the quarried subspaces is denoted by $\Delta$ and the error rate is by $\eps$.}
\label{table - related work}
\end{table}

\section{Main Contribution}
Our main technical result is an algorithm that gets a set $L$ of $n$ lines in $\REAL^d$, an integer $k \geq 1$, and computes an $\eps$-coreset (see Definition \ref{definition - coreset}) of size $dk^{O(k)}\log (n)/\eps^2$ for $L$ and every error parameter $\eps > 0$, in a near-linear running time in the number of data lines $n$, and polynomial in the dimensionality $d$ and the number $k$ of desired centers; See Theorem \ref{theorem - coreset offline main} and Theorem \ref{theorem - coreset streaming main} for details and exact bounds. 

Using this coreset with a merge-and-reduce technique, we achieve the following results:

\begin{enumerate}
  \item An algorithm that, during one pass, maintains and outputs a $(1 + \varepsilon)$-approximation to the $k$-line median of the lines seen so far; See Definition~\ref{definition - k mean}.
  
  \item A streaming algorithm that computes an $(1 + \varepsilon)$-approximation for the $k$-line median of a set $L$ of lines that may be distributed (partitioned) among $M$ machines, where each machine needs to send only $d^3k^{O(k)}\log^2n$ input lines to the main server at the end of its computation.
  
  \item Experimental results on 10 machines on Amazon EC2 Cloud \cite{amazoneec2} show that the algorithm performs well in practice, boost the performance of existing EM-heuristic \cite{Llyoodkmeans}. Moreover, open source code for all the algorithms and experiments is provided in \cite{my_code}.
\end{enumerate}

\chapter{Problem Statement} \label{section - problem statement}
\section{Preliminaries} \label{subsection - preliminaries}
From here and in the rest of the paper, the following holds:
\begin{itemize}
\item For an integer $n \geq 1$ we define $[n] = \br{1, \ldots, n}$.

\item We denote by $\Q_k=\br{Q\subseteq\REAL^d\mid |Q|=k}$ the family of all sets which are the union of $k$ points in $\REAL^d$.

\item We assume that we are given a function $D: \REAL^d \times \REAL^d \to \REAL$ and a constant $\rho>0$ such that $\dist(a,b) \leq \rho (\dist(a,c) + \dist(c,b))$ for every $a,b,c \in \REAL^d$.

\item Given a supremum $\displaystyle \sup_x \frac{f_1(x)}{f_2(x)}$, we assume that $f_2(x) \neq 0$.
\end{itemize}

\begin{definition} [weighted set]
A weighted set of lines is a pair $L'=(L,w)$ where $L$ is a set of lines in $\REAL^d$, and $w : L \to (0,\infty)$ is a function that maps every $\ell \in L$ to $w(\ell) \geq 0$, called the \emph{weight of $\ell$}. A weighted set $(L,1)$ where $1$ is the weight function $w : L \to \br{1}$ that assigns $w(\ell) = 1$ for every $\ell \in L$ may be denoted by $L$ for short.
\end{definition}

We will use the following definitions and observation due to Jubran et al. \cite{Jubran} to generalize our results for handling robust and other variants of distance functions. 

\begin{definition} [Non-decreasing funtion \cite{Jubran}]
for every pair of vectors $v = (v_1,\ldots, v_n)$ and 
$u = (u_1, \ldots ,u_n)$ in $\REAL^n$, we denote 
$v \leq u$ if $v_i \leq u_i$, for every 
$i \in [n]$. Similarly, $f : \REAL^n \to [0, \infty)$ is \emph{non-decreasing} if $f(v) \leq f(u)$ for every $v \leq u \in \REAL^d$.
\end{definition}
\begin{definition} [log-Lipschitz function \cite{Jubran}] \label{definition - log lipshitz fun}
Let $r,n \geq 1$ be two integers, $I$ be a subset of $\REAL^n$ and $h : I \to [0, \infty)$ be a non-decreasing function. Then $h(x)$ is $r$-$\mathrm{log}$-$\mathrm{Lipschitz}$ over $x \in I$, if for every 
$c \geq 1$ and $\displaystyle \ x \in I \cap \frac{I}{c}$, we have $h(cx) \leq c^rh(x)$. The parameter $r$ is called the $\mathrm{degree}$ of $h$.
\end{definition}

\begin{definition} [cost function \cite{Jubran}] \label{definition - jubran cost robust}
Let $X$ be a set called \emph{ground set}, $A=\br{a_1, \ldots, a_n} \subset X$ be a finite set and let $Q$ be a set called \emph{queries}. 
Let $\mathrm{dist}:X\times Q \to [0,\infty)$ be a function. Let $\lip:[0,\infty) \to [0,\infty)$ be an $r$-log-Lipschitz function, 
$f : [0,\infty)^n \to [0,\infty)$ be an  $s$-log-Lipschitz function. For every $q\in Q$ we define
$$\cost_{\mathrm{dist},\lip,f}(A,q)=f(\lip(\mathrm{dist}(a_1,q)), \ldots, \lip(\mathrm{dist}(a_n,q))).$$
\end{definition}

\begin{definition}[distance]
The Euclidean distance between a pair of points is denoted by the function $\dist:\REAL^d \times \REAL^d \rightarrow [0,\infty)$, s.t. for every $x,y \in \REAL^d$ we have $\dist(x,y) = \norm{x-y}_2$. For every set $X \subseteq \REAL^d$ and a point $x \in \REAL^d$, we define the distance from $X$ to $x$ by $\dist(X,x) = \inf_{q \in X} \dist(q,x)$, and for every set $Y \subseteq \REAL^d$, we denote the distance from $X$ to $Y$ by $\dist(X, Y) = \inf_{(x,y)\in X \times Y} \dist(x, y)$. 
\end{definition}

\begin{definition}[cost] \label{definition - cost}
For every set $P \subseteq \REAL^d$ of $k$ points and a weighted set of lines $L' = (L,w)$ in $\REAL^d$, we denote the sum of weighted distances from $L$ to $P$ by $\cost(L',P)=\sum_{\ell \in L} w(\ell)\dist(\ell, P)$.
\end{definition}

\begin{definition}[projection $\pi$]
For every two sets $X,Y \subseteq \REAL^d$, we denote $\proj(X,Y) \in \arginf_{y \in Y} \dist(X,y)$ to be the closest point in $Y$ to $X$, ties are broken arbitrary.
\end{definition}

\begin{definition} [closest] \label{definition - closest}
Let $X_1,\ldots,X_n \subseteq \REAL^d$ be $n$ lines in $\REAL^d$, $X=\br{X_i \mid i \in [n]}$, $m \geq 1$ be an integer and $B \subseteq \REAL^d$ be a finite set of points. We define 
$$
\closest(X,B,m) \in \argmin_{X' \subseteq X \abs{X'}=m} \cost(X',B),
$$ 
as the $m$ closest lines to $B$ in $X$. Ties are broken arbitrarily. 
\end{definition}

\begin{definition} [query space \cite{zahimsc}] \label{definition - k query space}
Let $Y$ be a set called \emph{query set} and $P' = (P,w)$ be a weighted set. 
Let $f : P \times Y \to [0,\infty)$ be a function called a \emph{kernel function} and 
$\loss : \REAL^d \to [0, \infty)$ be a function that assigns a non-negative real number for every real vector. The tuple 
$(P',Y,f,\loss)$ is called a 
\emph{query space}. For every $y \in Y$ we define the overall fitting error of $P'$ to $y$ by
$$
f_{\loss}(P',y) 
\coloneqq \loss \of{\of{w(p)f(p,y)}_{p \in P}}
= \loss\of{w(p_1)f(p_1,y),\ldots, u(p_{\abs{P}})f(p_{\abs{P}},y)}.
$$
\end{definition}

\section{Problem Statement and Theoretical Result}
In the familiar $k$-median clustering problem, the input is a set $P$ of $n$ points in $\REAL^d$, and the goal is to compute a set $C$ of $k$ centers (points) in $\REAL^d$, that minimizes the sum of distances over each point $p \in P$ to its nearest center in $C$, i.e.
$$
C \in \argmin_{C' \subseteq \REAL^d, \abs{C'}=k} \sum_{p \in P} \min_{c' \in C'} \norm{p-c'}.
$$
A natural generalization of the $k$-median problem is to replace the input set of points $P$ by a set $L$ of $n$ lines in $\REAL^d$.
\begin{definition}[$k$-median for lines] \label{definition - k mean}
Let $L'=(L',w)$ be a weighted set of lines in $\REAL^d$ and $k \geq 1$ be an integer. A set $P^* \subseteq \REAL^d$ is a \emph{$k$-median} of $L'$ if it minimizes $\cost(L',P)$ over every set $P$ of $k$ points in $\REAL^d$.

\end{definition}
In this thesis, for every weighted set $L'=(L,w)$ of finite number of lines in $\REAL^d$, we aim to compute a weighted set of lines $C'=(C,u)$ where $C \subseteq L$, which is a small summarization of $L'$, in the sense that $\cost(C',P)$ approximates $\cost(L', P)$ for every set $P \subseteq \REAL^d$ of $k$ points. This enables us to boost performance of common state-of-the-art heuristics, since we apply them on a much smaller set of lines (the coreset sample). 

Alternatively, we can compute a PTAS for the k-line median of $L'$, i.e., a $(1+\eps)$-approximation in time that is near-linear in $n=|L|$.

\paragraph{Robustness to Outliers.} More generally, we suggest solution of the $k$-median with outliers resistance problem, that are robust to outliers, i.e., lines in $L$ that are farthest from the desired set $P$ of $k$ centers. This problem is much harder than the original $k$-median problem since we need to compute the centers and outliers simultaneously, which is usually harder problem.

\begin{definition}
Let $v \in \REAL^n$ and $j \in [n]$. Let $\mathrm{smallest}(v, j) = (v(1), \ldots, v(j)) \in \REAL^j$ denote the smallest $j$ entries in $v$, ties broken arbitrary. For every $z > 0$ we denote by $\norm{v}_{z,j} = \norm{\mathrm{smallest}(v, j)}_z$.
\end{definition}

The corresponding loss function for every set $P$ of $k$ points is then
$$
\norm{(\dist(\ell_1,P), \ldots, \dist(\ell_n,P))}_{1,n-j}
.$$

\paragraph{Coreset.} Coreset is a problem dependent data summarization. The definition of coreset is not consistent among papers. In this thesis, the input is usually a set of lines in $\REAL^d$, but for the streaming case in Section~\ref{section - Coreset for Streaming Data} we compute coreset for union of (weighted) coresets and thus weights will be needed. We use the folowing definition of Feldman and Kfir \cite{zahimsc}.

\begin{definition} [$\eps$-coreset \cite{zahimsc}] \label{definition - coreset}
For an approximation error $\eps > 0$, the weighted set $S' = (S,u)$ is called an \emph{$\eps$-coreset} for the query space $(P',Y,f,\loss)$, if $S \subseteqÂ„ P, u : S \to [0, \infty)$, and for every $y \in Y$ we have
$$
(1 - \eps)f_{\loss}(P', y) \leq f_{\loss}(S',y) \leq (1 + \eps)f_{\loss}(P', y).
$$
\end{definition}

\begin{theorem} [coreset for $k$-line median] \label{theorem - coreset offline main}
Let $L'=(L,w)$ be a weighted set of $n$ lines in $\REAL^d$, $k \geq 1$ be an integer, $\eps,\delta\in(0,1)$ and $m >1$ be an integer such that
$$
m \geq \frac{cd^2k\log^2(k) \log^2(n)+\log(1/\delta)}{\eps^2},
$$ 
for some universal constant $c>0$ that can be determined from the proof, and $\Q_k=\br{B\subseteq\REAL^d\mid\abs{B}=k}$. Let $(S,u)$ be the output of a call to $\coreset(L,k,m)$; see Algorithm \ref{algorithm - coreset}. Then, with probability at least $1-\delta$, $(S,u)$ is an $\eps$-coreset for the query space $(F^*_{L'} ,\Q_k, \dist, \norm{\cdot}_1)$, where $F^*_{L'}$ is defined as in Corollary~\ref{corollary - the range space is O(dklog k)}. Moreover, $(S,u)$ can be computed in time
$$
O\of{d^2 n\log (n) k \log k} + ndk^{O(k)}
.$$
\end{theorem}

\begin{theorem} \label{theorem - coreset streaming main}
Let $stream=\br{\ell_1,\ell_2,\ldots}$ be a stream of lines in $\REAL^d$, and let $n>0$ denote the number of lines seen so far in the stream. Let $k \geq 1$ be an integer, $c>0$ be a suffice large constant, $\eps, \delta \in (0,1)$, and let 
$$
m \geq \frac{cd^2k\log^2(k)\log^2(n)\log(e/\delta)}{\eps^2}.
$$ 
For every $h \geq 1$ we define
\begin{equation} \label{equation - hlf(h) value}
\hlf(h) = \frac{ch^9m^3}{\ln^3(n)}.
\end{equation}
Let $S'_1,S'_2,\ldots$ be the output of a call to $\textsc{Streaming-Coreset}(stream, \eps/6, \delta/6,\coreset,\hlf)$, where a call to $\coreset((Q,w),\eps,\delta)$ returns a weighted set $S'=(S,u)$ whose overall weight is $\sum_{p \in S}u(p) = \sum_{p \in Q}w(p)$; See Alg. \ref{algorithm - streaming coreset}. Then, with probability at least $1-\delta$, the following hold. For every $n \geq 1$:\\
\\$(i)$ (Correctness) $S'_n$ is an $\eps$-coreset of $(L_n ,\Q_k, \dist, \norm{\cdot}_1)$, where $L_n$ is the first $n$ lines in $stream$.\\
\\$(ii)$ (Size)
$$
\abs{S'_n} \in O\of{\frac{m^3}{\ln^3n}}
.$$
$(iii)$ (Memory) there are
$$
b \in O\of{m^3}
$$
lines in memory during the streaming.\\
\\$(iv)$ (Update time) $S'_n$ is outputted in additional 
$$
\displaystyle t \in O\of{d^2 b\log (b) k \log k} + bdk^{O(k)}
$$
time after $S'_{n-1}$.\\
\\$(v)$ (Overall time) $S'_n$ is computed in $nt$ time.

\end{theorem}

\chapter{Algorithms}
In this section we present a bi-criteria approximation algorithm for the $k$-line median problem. The main pseudo code is in Algorithm~\ref{algorithm - a,b approx} which calls Algorithm 1 as a sub-routine.

\paragraph{Overview of Algorithm~\ref{algorithm - 4 approx}.} The algorithm gets as an input a set $L$ of $n$ lines in $\REAL^d$, and returns a set $G \subseteq \REAL^d$ of ${n \choose 2}$ points that we call a centroid-set. This centroid set contains an approximated solution for the $k$-median of $L$ as proved in Lemma~\ref{lemma - 4 approx}. This is by simply iterating over every pair $(\ell,\ell')$ of input lines in $L$, and computing the closest point to $\ell'$ that is contained in $\ell$.

%algorithm - 4 approx
\begin{algorithm} 
	\caption{$\constantapprox(L)$\label{algorithm - 4 approx}}
	\begin{tabbing}
		\textbf{Input:} \quad\quad\= A finite set $L$ of $n$ lines in $\mathbb{R}^d$. \\
		\textbf{Output:} \>A set $G \subseteq \REAL^d$ of $O(n^2)$ points that satisfies Lemma~\ref{lemma - 4 approx}.
		
	\end{tabbing}
	\vspace{-0.3cm}	
	\nl \For{every ${\ell}\in L$}{
		
		\nl \For{every $\ell' \in L \setminus \ell$}{

			\nl Compute $q(\ell, \ell') \in \argmin_{x \in \ell} \dist(\ell',x)$ \\\quad\tcp{the closest point on $\ell$ to $\ell'$. Ties broken arbitrarily.}			
			
		}		
				
		\nl $Q(\ell)\coloneqq\br{q(\ell,\ell')\mid \ell'\in L\setminus \br{\ell}}$
		
		\nl $G \coloneqq \bigcup_{\ell \in L} Q(\ell)$
		
	}
	
	\nl \Return $G$
\end{algorithm}

\clearpage
\paragraph{Overview of Algorithm~\ref{algorithm - a,b approx}.} The input to the algorithm is a set $L$ consist of $n$ lines in $\REAL^d$ and a positive integer $m \geq 1$. In each iteration of the algorithm it picks a small uniform sample $S$ of the input in Line 4, compute their centroid-set $G$ using a call to Algorithm~\ref{algorithm - 4 approx} in Line 5, and add them to the output set $B$ in Line 6. Then,  in Line 7, a constant fraction of the closest lines to $G$ is removed from the input set $L$. The algorithm then continues recursively for the next iteration, but only on the remaining set of lines until almost no more lines remain. The output is the resulting set $B$.

%algoritm - a,b approx
\begin{algorithm}
	\caption{$\bicriteria(L, m)$\label{algorithm - a,b approx}}
	\begin{tabbing}
		\textbf{Input:} \quad\quad\= A set $L$ of $n$ lines in $\mathbb{R}^d$, and an integer $m \geq 1$.\\
		\textbf{Output:} \> A set $B\subseteq\REAL^d$ which is, with probability at least $1/2$, an $(\alpha,\beta)$-approximation \\ \>for the $k$-median of $L$, where $\alpha \in O(1)$ and $\beta=O\of{m^2 \log n}$.
\\\>\quad\tcp{See Definition~\ref{definition - a,b approx} and Theorem~\ref{theorem - the output of a,b approx is indeed a,b approx}.}
	\end{tabbing}
	\vspace{-0.3cm}
	
	\nl $B \coloneqq \emptyset$	
	
	\nl $X \coloneqq L$	
		
	\nl \While{ $ |X| > 100$ } {
				
		\nl Pick a sample $S$ of $\abs{S} \geq m$ lines, where each line $\ell \in S$ is sampled i.i.d and uniformly at random from $X$.
		
		\nl $G \coloneqq \constantapprox(S)$
		
		\nl $B \coloneqq B \cup G$
		
		\nl $X' \coloneqq$ the closest $7\abs{X}/11$ lines in $X$ to $G$. Ties broken arbitrarily.
				
		\nl $X \coloneqq X\setminus X'$
		
	}	
	
	\nl \Return $B$
	
\end{algorithm}

\begin{figure} [h]
	%\centering
	\includegraphics[scale=0.44]{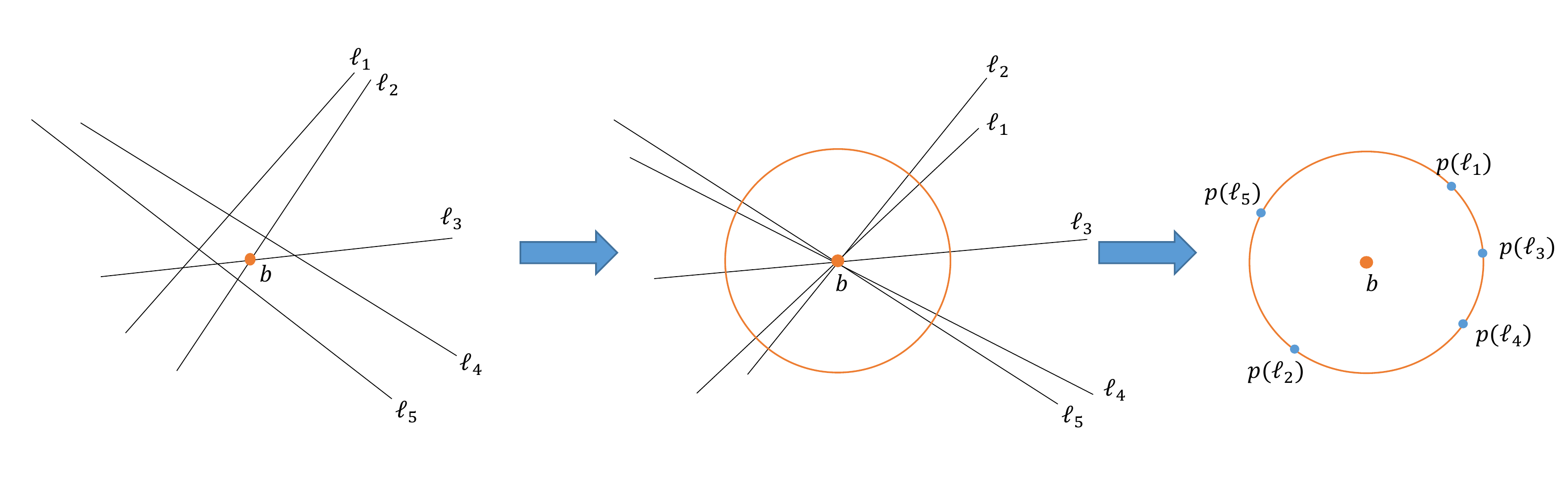}
		\caption{\textbf{Example of running Alg.~\ref{algorithm - sensitivityoftranslatedlines} for lines on the plane.} \textbf{(left)} The input is a set $L = \br{\ell_1, \ldots, \ell_5}$ of lines and a single center $b$ on the plane. \textbf{(middle)} Every line is translated to $b$ and we define unit sphere $\mathbb{S}^{d-1}$ around it. \textbf{(right)} The intersection point $p(\ell_i)$ of each line $\ell_i$ with the unit sphere $\mathbb{S}^{d-1}$ around $b$ is computed.}		
		\label{figure - lines sensitivity}
\end{figure}

\clearpage

\paragraph{Overview of Algorithm \ref{algorithm - sensitivityoftranslatedlines}.}
The input is a set $L$ of lines in $\REAL^d$, a point $b\in\REAL^d$ and an integer $k \geq 1$ for the number of desired centers. This procedure is called from Algorithm~\ref{algorithm - coreset}, where $b$ is an approximation to the $1$-median of $L$. The output is a function $s$ that maps every line $\ell \in L$ to $[0, \infty)$, and is being used in Lemma~\ref{theorem - total sensitivity can be computed given a,b-approx} proof. We first define in Line 1 a unit sphere $\mathbb{S}^{d-1}$ that is centered around $b$.
Next, in Line 3 for each line $\ell \in L$ we define $\ell' \subseteq \REAL^d$ to be the translation of the line $\ell$ to $b$.
In Lines 4--5, we replace every line $\ell'$ with one of its two intersections with $\mathbb{S}^{d-1}$, and define the union of these points to be the set $Q$.
In Line 6 we call the sub-procedure $\sensitivityforweightedcenters$ that is described in~\cite{outliers-resistance}. This procedure returns the sensitivities of the query space of $k$-weighted centers queries on $Q$. As stated in Theorem~\ref{theorem - total sensitivity of k weighted points theorem}, the total sensitivites of this coreset is $k^{O(k)}\log n$.
Finally, in Line 7, we convert the output sensitivity $s(p)$ of each point $p$ in $Q$ to the output sensitivity $s(\ell)$ of the corresponding line $\ell$ in $L$.

%algorithm - sensitivityoftranslatedlines
\begin{algorithm} 
	\caption{$\sensitivityoftranslatedlines(L,b,k)$\label{algorithm - sensitivityoftranslatedlines}}
	\begin{tabbing}
		\textbf{Input:} \quad\quad\=A set $L$ of $n$ lines in $\mathbb{R}^d$, a point $b \in \REAL^d$ and integer $k \geq 1$.\\
		\textbf{Output:} \>A (sensitivity) function $s : L \to [0,\infty)$.
	\end{tabbing}
	\vspace{-0.3cm}
	
	\nl $\mathbb{S}^{d-1} \coloneqq\br{x\in \REAL^d\mid \norm{x-b}=1}$ \tcp{the unit sphere that is centered at $b$.} 				
	
	 \For{$\ell \in L$}{		
			
		$\ell' \coloneqq$ the line $\br{x-b\mid x\in\ell}$ that is parallel to $\ell$ and intersects $b$
\tcp{see Fig.~\ref{figure - lines sensitivity}.}
		
		 $p(\ell') \coloneqq$ an arbitrary point in the pair $\ell' \cap \mathbb{S}^{d-1}$
		
	}	

	$Q \coloneqq Q \br{p(\ell') \mid \ell \in L}$					
	
	$u \coloneqq \sensitivityforweightedcenters(Q,2k)$ \tcp{see algorithm overview.}
	
Set $s:L \rightarrow [0,\infty)$ such that for every $\ell \in L$
\[
	s(\ell) \coloneqq u \of{p(\ell')}.
\]	
	
 \Return $s$
\end{algorithm}

\clearpage

\paragraph{Overview of Algorithm \ref{algorithm - coreset}.}
The algorithm gets a set $L$ of lines in $\REAL^d$, an integer $k \geq 1$ for the number of desired centers and a positive integer $m \geq 1$ for the coreset size, and returns an $\eps$-coreset for $L$; See Definition~\ref{definition - coreset}.
In Line 2 a small set $B$ of points that approximate the $k$-median of $L$ is computed via a call to \bicriteria. In Line 3 the lines in $L$ are clustered according to their nearest point in $B$, and in Line 5 the sensitivity of the lines in the cluster $L_b$ are computed for each center $b \in B$. In the second "for" loop between Lines 8--10 we set the sensitivity of each line to be the sum of the scaled distance of the line to its nearest center $b$ (translation), and the sensitivity $s_b$ that measure its importance with respect to its direction (rotation). Here, scaled distance means that the distance is divided by the sum of distances $\cost(L,B)$ over all the lines in $L$. The If statement in Line 7 is used to avoid division by zero. 
In Line 12 we pick a random sample $S$ from $L$, where the probability of choosing a line $\ell$ is proportional to its sensitivity $s(\ell)$. In Line 13 we assign a weight to each line, that is inverse proportional to the probability of sampling it. The resulting weighted set $(S,u)$ is returned in Line 14.

%algorithm - coreset
\begin{algorithm} 
	\caption{$\coreset(L,k,m)$\label{algorithm - coreset}}
	\begin{tabbing}
		\textbf{Input:} \quad\quad\= A finite set $L$ of lines in $\mathbb{R}^d$, number $k \geq 1$ of centers and the coreset size $m \geq 1$.\\
		\textbf{Output:} \>A weighted set (``coreset'') $(S,u)$ that satisfies Theorem \ref{theorem - coreset offline main}.
	\end{tabbing}
	\vspace{-0.3cm}
	
	\nl $j \coloneqq cdk \log_2 k$, where $c$ is a sufficient large constant $c>0$ that can be determined from the proof of Theorem~\ref{theorem - the output of a,b approx is indeed a,b approx}.	
	
	\nl $B \coloneqq \bicriteria\of{L, j}$ \tcp{see Algorithm~\ref{algorithm - a,b approx}}
	
	Compute a partition $\br{L_b\mid b\in B}$ of $L$ such that $L_b$ is the set (cluster) of lines that are closest to the point $b\in B$. Ties broken arbitrarily.		
		
	\For{every ${b}\in B$}{		
	
	 $s_b \coloneqq \sensitivityoftranslatedlines(L_b,b,k)$
\\\tcp{ the sensitivity of each line $\ell \in L_b$ that was translated onto $b$; see Algorithm~\ref{algorithm - sensitivityoftranslatedlines}}
	
	}
			
	 \For{every ${b}\in B$ and $\ell \in L_b$}{	
	
		\eIf{$\cost(L,B) > 0$}{
   $s(\ell) \coloneqq \displaystyle\ \frac{\dist(\ell,b)}{\cost(L,B)}\ +  2 \cdot s_b(\ell)$ 	
   }{
   $s(\ell) \coloneqq s_b(\ell)$ 	
  }

		  $\prob(\ell) \coloneqq \displaystyle\ \frac{s(\ell)}{\sum_{\ell' \in L}s(\ell')}\ $
		
	}		
	
	Pick a sample $S$ of at least $m$ lines from $L$, where each line $\ell \in L$ is sampled i.i.d. with probability $\prob(\ell)$.
	
	Set $u:S \rightarrow [0,\infty)$ such that for every $\ell \in S$	
\[
	u(\ell) \coloneqq \frac{1}{\abs{S}\prob(\ell)}.
\]
	
	\Return $(S,u)$
	
\end{algorithm}

\chapter{Analysis}
In this chapter we will prove and analyze the correctness and running time of each one of the aforementioned algorithms in the last chapter.

\section{Algorithm~\ref{algorithm - 4 approx}: \constantapprox}
\subsection{Analysis of Algorithm~\ref{algorithm - 4 approx}: \constantapprox}
Algorithm~\ref{algorithm - 4 approx} gets a set $L$ of $n$ lines in $\REAL^d$ and returns a set $G \subseteq \REAL^d$ of ${n \choose 2}$. In this section we prove that $G$ contains a constant factor approximation for the $k$-median of $L$, as stated in Lemma~\ref{lemma - 4 approx}.

A main observation that we use is that the distance from a point $p$ to a line $A$ is the same as its weighted distance to a point $H$. Here, $H$ and the weight depends only on the two lines and not on the point $p$. This is easy to see for the case of two lines on the plane. Less intuitive is the fact that this holds also for 2-lines in 3-dimensional space, and is described as follows.

\begin{lemma} [\cite{feldmanphd}] \label{theorem - weighted point for every subspace Feldman}
Let $f$ be a $j$-dimensional affine linear subspace in $\REAL^d$, for some $j \in [d-1]$. Let $g$
be an affine linear subspace in $\REAL^d$ of any dimension such that $g$ does not contain a translation of $f$.
Let $(v,\omega)$ denote the affine linear subspace $v$ and the constant $\omega > 0$ which are the output of the
algorithm $\textsc{Weighted-Flat}(f, g)$; see Fig 4.7 in \cite{feldmanphd}. Then $v$ is an affine linear subspace of dimension at most $j - 1$, and for each $p \in f$ we have
$$
\dist(g, p) = w \cdot \dist(v, p)
.$$
\end{lemma}

The following lemma proves that for every line $\ell \in L$, the output set $G \subseteq \REAL^d$ of a call to $\constantapprox(L)$ contains a set $P \subseteq G$ that approximates the distance to $\ell$ from every set of $k$ points in $\REAL^d$. Note that this claim is much stronger than the claim that, say, the sum of distances is approximated by $P$. We need this stronger claim for the robust and generic approximation that is stated in Theorem~\ref{theorem - 4 approx is robust}. Recall that $\rho>0$ is defined such that $\dist(a,b) \leq \rho (\dist(a,c) + \dist(c,b))$ for every $a,b,c \in \REAL^d$, we define the following lemma.
\begin{lemma} \label{lemma - 4 approx}
Let $L$ be a set of $n$ lines in $\REAL^d$, $k \geq 1$ be an integer, and let $G \subseteq \REAL^d$ be the output of a call to $\constantapprox(L)$; See Alg.~\ref{algorithm - 4 approx}. Then, for every set $P\subseteq\REAL^d$ of size $\abs{P}=k$, there is a set $P''\subseteq G$ of $k$ points such that
\[
\forall \ell\in L: \dist(\ell,P'') \leq 4 \rho^2 \cdot \dist(\ell,P).
\]
Moreover, $G$ can be computed in $O\of{d^2n^2}$ time.
\end{lemma}
\begin{proof}
Let $P\subseteq\REAL^d$ be a set of $k$ points. Let $\ell\in L$, and $p\in P$ denote a closest point in $P$ to $\ell$. 
Let $\label{p'} p'=\proj\of{p,\bigcup_{\ell''\in L}\ell''}$ denote the projection of $p$ onto its closest line $\ell'$ in $L$, and let $P'=\br{p'\mid p\in P}$ denote the union over every $p \in P$; See Fig.~\ref{figure - example of translation of the center points}. Hence,

\begin{align}
\dist(\ell, p')  
& = \dist(\proj(p',\ell), p')   \nonumber \\
& \leq \dist(\proj(p,\ell), p')   \label{equation - first_translation - a}\\
& \leq \rho \of{\dist(\proj(p,\ell),p) + \dist(p,p')}   \label{equation - first_translation - b}\\
& = \rho\of{\dist(\ell,p) + \dist(p,p')} \nonumber \\
& \leq 2 \rho \dist(\ell,p) \label{equation - first_translation - c} \\
& = 2 \rho \dist(\ell,P) \label{equation - first_translation - d}
\end{align}
where~\eqref{equation - first_translation - a} is by the definition of $\proj(p',\ell)$,~\eqref{equation - first_translation - b} is by the triangle inequality, and~\eqref{equation - first_translation - c} holds since $p'$ is defined to be the closest point to $p$ in $\bigcup_{\ell'\in L}\ell'$.

%figure - example of translation of the center points
\begin{figure}[h]
		\caption{\textbf{$4$-approximation for the $k=2$-median of an input set $L=\br{\ell_1, \ldots, \ell_5}$ of $n=5$ lines in the plane.} (left) The (unknown, optimal) $2$-median $P=\br{p_1,p_2}$ is projected on $L$ to obtain $2$-approximation. (middle) Each projected point in $P'=\br{p'1_,p'_2}$ is translated to the nearest intersection point. (right) The resulting set $P''= \br{p''_1,p''_2}$ is a $4$-approximation that is contained in the output $G \subseteq \REAL^2$ of Algorithm~\ref{algorithm - 4 approx}.}
		\centering
		\includegraphics[width=16cm]{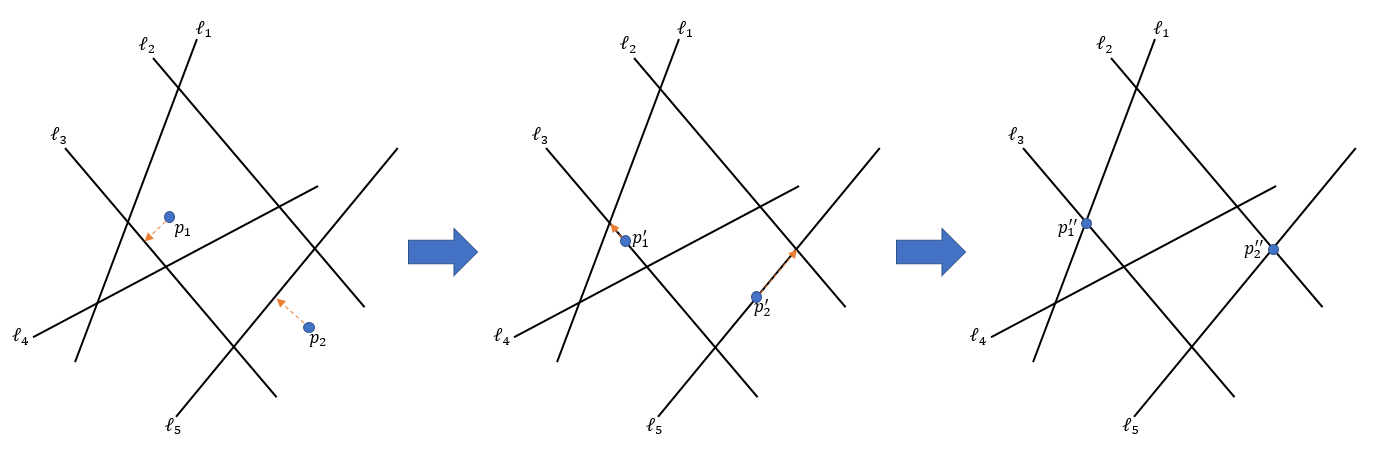}
		\label{figure - example of translation of the center points}
\end{figure}

Let $\ell'$ be the closest line to $p$ in $L$ that satisfies $p'=\proj(p,\ell')$. Recall that $Q(\ell')$ as defined in Line 4 in Alg.~\ref{algorithm - 4 approx} is the union of the $n-1$ closest points in $\ell'$ to each $\ell \in L \setminus\br{\ell'}$. Let $p''\in\argmin_{q\in Q(\ell')}\dist(p',q)$ denote a closest point to $p'$ in $Q(\ell')$, and $P'' = \bigcup_{p' \in P'}\argmin_{q\in Q(\ell')}\dist(p',q)$ denote the union over every $p' \in P'$; See Fig.~\ref{figure - example of translation of the center points}. 

Substiuting $j=1,f=\ell'$ and 
$g=\ell$ in Lemma~\ref{theorem - weighted point for every subspace Feldman} and noting that $p'' \in \ell'$ yields that there is $\omega\geq 0$ and $v \in \REAL^d$ such that

\begin{equation} \label{equation - distance is a weighted point}
\forall t \in \ell' : \dist(\ell, t)= \omega \cdot \dist(v, t).
\end{equation}
Note that Lemma~\ref{theorem - weighted point for every subspace Feldman} consider the case where $\ell$ and $\ell'$ are not parallel one to each other. However, elementary geometry implies that~\eqref{equation - distance is a weighted point} still holds, since in this case $\omega=1$ and $\dist(\ell, \ell') = \dist(\ell, v)$ for any $v \in \ell'$. This gives us
\begin{align}
\proj(v, \ell')
& \in \argmin_{t \in \ell'}D(v,t) \label{equation - expand of proj(v, ell') in Q(ell') - a} \\ 
& = \argmin_{t\in \ell'}D(\ell,t)/w \nonumber \\
& = \argmin_{t\in \ell'}D(\ell,t) 
\subseteq Q(\ell') \label{equation - proj(v, ell') in Q(ell')}
\end{align}
where~\eqref{equation - expand of proj(v, ell') in Q(ell') - a} is by ~\eqref{equation - distance is a weighted point}. Hence,
\begin{align} 
\dist(\ell, P'')
& \leq \dist(\ell, p'') \nonumber \\
& = \omega \cdot \dist(v,p'')   \label{equation - second_translation - a}\\
& \leq \omega \rho \cdot \big( \dist(v,p')+\dist(p',p'')\big)   \label{equation - second_translation - b}\\
& \leq \omega \rho \cdot \big( \dist(v,p')+\dist(p',\proj(v,\ell')) \big) \label{equation - second_translation - c} \\
& \leq 2 \omega \rho \cdot \dist(v,p') \label{equation - second_translation - d} \\
& = 2 \rho \dist(\ell,p'), \label{equation - second_translation - e}
\end{align}
where \eqref{equation - second_translation - a} follows by substituting $t = p''$ in \eqref{equation - distance is a weighted point}, \eqref{equation - second_translation - b} holds by the approximated triangle-inequality; See Section~\ref{subsection - preliminaries}. Inequality \eqref{equation - second_translation - c} holds by combining \eqref{equation - proj(v, ell') in Q(ell')} and the definition of $p''$ as the closest point to $p'$ in $Q(\ell')$, and~\eqref{equation - second_translation - d} holds since $p'\in\ell'$ and by the Pythagorean Theorem its distance to any point $v \in \REAL^d$ is larger than the projection of $v$ on $\ell'$, i.e., \\$\dist(p',\pi(v,\ell'))\leq \dist(p',v)$; See Fig.~\ref{figure - the cost to tilde P is a factor 2 from the cost to P'}. This proves Lemma~\ref{lemma - 4 approx} since 

$$
\dist(\ell, P'') 
\leq 2 \rho \dist(\ell, p')
\leq 4 \rho^2 \dist(\ell, P),
$$
where the first inequality is by~\eqref{equation - second_translation - e} and the second is by~\eqref{equation - first_translation - d}.

The running time of Algorithm~\ref{algorithm - 4 approx} is dominated by Lines 1--2, that is executed $O(n^2)$ times. Each time the distance between a pair of lines in $\REAL^d$ is computed. This can be done in $O(d^2)$ time via solving a $d^2$-degree polynomial equation. Hence, the overall running time of Alg. 1 is $O(d^2)\cdot n(n-1)\in O(d^2n^2)$.
\end{proof}

%figure - the cost to tilde P is a factor 2 from the cost to P'
\begin{figure}[h]
		\caption{\textbf{Illustration of the proof of Lemma \ref{lemma - 4 approx} for $d=3$ dimensional space.}\\(i) For every pair $\ell$ and $\ell'$ of lines, there is a point $v$ and a weight (scalar) $w>0$ that satisfies the following: the distance of every $p''\in\ell$ to the line $\ell$ is the same as its weighted distance to the point $v$. (ii) $p'$ is the projection of $p$ onto $\ell'$, and its distance to $v$ is larger than $\dist(v,\ell')$.}
		\centering
		\includegraphics[height = 6 cm]{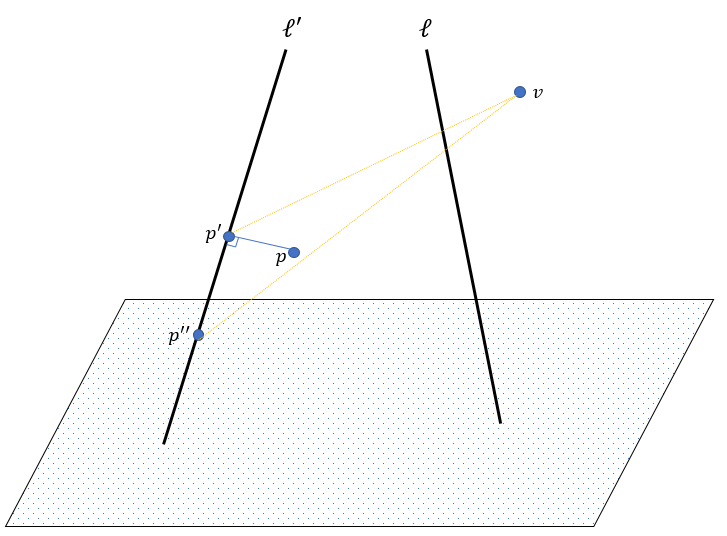}
		\label{figure - the cost to tilde P is a factor 2 from the cost to P'}
\end{figure}

The following observation is based on Definition~\ref{definition - jubran cost robust}. It states that if there is a center that approximates the distance to the optimal solution per point, then it also approximates e.g. max, sum, or sum of squared distances to all the points.

\begin{observation} [\cite{Jubran}] \label{observation - Jubran lipshitz}
Let $\cost_{\mathrm{dist},\lip,f}(A,q)=f(\lip(\mathrm{dist}(a_1,q)), \ldots, \lip(\mathrm{dist}(a_n,q)))$ be defined as in Definition \ref{definition - jubran cost robust}. Let $q^*,q' \in Q$ and let $c \geq 1$. If $\lip(\mathrm{dist}(a_i,q')\leq c\cdot \lip(\mathrm{dist}(a_i,q^*)$ for every $i \in [n]$, then
\begin{equation} \label{equation - jubran cost dist lip f}
\cost_{\mathrm{dist},\lip,f}(A,q') \leq c^{rs} \cdot \cost_{\mathrm{dist},\lip,f}(A,q^*)
.
\end{equation}
\end{observation}

The following result is an example application of combining Lemma~\ref{lemma - 4 approx} and Observation~\ref{observation - Jubran lipshitz}. It states that the output of Algorithm~\ref{algorithm - 4 approx} also contains a constant factor approximation for solving $k$-median of lines with outliers. Even if the points of this $k$-median have additional constraints (such as being subset of the input, restricted zones, or far from specific areas).
This is although the algorithm uses only Euclidean distances. The same output set also contains an approximation to the $k$ centers that minimizes sum of squared distances, or maximum distance.

In the next chapters the approximation factor is reduced to $1+\eps$ and the running time is reduced to be near-linear in $n$. This is by applying Algorithm~\ref{algorithm - 4 approx}  only on small subsets of the input and using it as a building block for computing coresets. 

Recall that $\cost(L',P)$ is the sum of weighted distances from a set of lines $L'$ to a set of points $P$, and $\closest(L',P,m) \in \argmin_{X' \subseteq X \abs{X'}=m} \cost(L',P)$ denote the sum of weighted distances from $P$ to its $m$ closest lines in $L'$; see Definitions~\ref{definition - cost} and~\ref{definition - closest}, respectively.

\begin{theorem} [$k$-median approximation with outliers] \label{theorem - 4 approx is robust}
Let $L = \br{\ell_1, \ldots, \ell_n}$ be a set of $n$ lines in $\REAL^d$, $k \geq 1$ be an integer and $G \subseteq \REAL^d$ be an output of a call to $
\constantapprox(L)$; See Alg.~\ref{algorithm - 4 approx}. Then for every integer $m \in [n-1]$ and a set $P \subseteq \REAL^d$ of $k$ points, there exists a set $P'' \subseteq G$ of $\abs{P''}=k$ points that robustly approximates the sum of distances from $P$ to $L$ up to a constant factor, i.e.,

\begin{equation}
\cost\of{\closest\of{L,P'',m}, P''} \leq 4 \rho^2 \cdot \cost\of{\closest\of{L,P,m}, P} .
\end{equation}
\end{theorem}
\begin{proof}
Let $m\in [n-1]$ and $P \subseteq \REAL^d$, $\abs{P}=k$. By Lemma~\ref{lemma - 4 approx}, there exists a set $P'' \subseteq G, \abs{P''}=k$ such that
\begin{equation} \label{equation - 4 approx}
\forall \ell \in L: \dist(\ell,P'') \leq 4 \rho^2 \cdot \dist(\ell,P).
\end{equation}
Let $X$ be a set of all the lines in $\REAL^d$, $Q=\br{Y\subseteq\REAL^d\mid |Y|=k}$ be the union over every set of $k$ points in $\REAL^d$, $A=L, \mathrm{dist}=\dist, \lip(x)=x$ for every $x>0$. Let
$$
f(x_1,\ldots,x_n)=\min_{\substack{M \subseteq \br{x_1, \ldots, x_n} \\\abs{M}=m}}\sum_{x \in M} x
$$ denote the sum of the $m$ smallest numbers in $\br{x_1, \ldots, x_n}$, $c=4\rho^2$ 
and the two sets \\$\displaystyle \ Q^* \in \argmin_{P^* \subseteq \REAL^d, \abs{P^*}=k} \cost(L,P^*)$ and $Q' = P''$. Substituting these in Observation~\ref{observation - Jubran lipshitz}, yields

\begin{align} 
\cost\of{\closest\of{L,P'',m}, P''}
& = f\of{\dist(\ell_1, P''), \ldots, \dist(\ell_n, P'')} \nonumber \\
& \leq c^{rs} \cdot f\of{\dist(\ell_1, P), \ldots, \dist(\ell_n, P)} \label{equation - k means with outliers a} \\
& = c^{rs} \cdot \cost\of{\closest\of{L,P,m}, P} \nonumber \\
& = 4 \rho^2 \cdot \cost\of{\closest\of{L,P,m}, P}, \label{equation - k means with outliers b} 
\end{align}
where~\eqref{equation - k means with outliers a} holds by~\eqref{equation - jubran cost dist lip f} and~\eqref{equation - k means with outliers b} holds since $\lip$ and $f$ are both $1$-$\log$-$\mathrm{Lipschitz}$ functions, i.e. $r=s=1$, which proves the theorem.
\end{proof}

\subsection{EM Algorithm for $k$-Line Mean}

A natural competitor for our $k$-line median constant factor approximation is the EM-algorithm, adopted for the special case of the $k$-line mean problem. We apply it for the sum of squared distances below.

\begin{theorem}
Let $L$ be a set of $n$ lines in $\REAL^d$. Then we can compute 
\[
\min_{x\in\REAL^d}\sum_{\ell\in L}\dist^2(x,\ell)
\]
in $O(nd^2)$ time.
\end{theorem}

\begin{proof} 
Let
$$
f(x) = \sum_{i \in [n]} \dist^2(x,\ell_i).
$$
For every $i \in [n]$, we denote each line $\ell_i \in L$ by the set of points $\br{t_ia_i-d_i \mid a_i,d_i \in \REAL^d, t_i \in \REAL}$. Let the line $\ell_i \in L$ and $x \in \REAL^d$, from the algebraic definition of distance between $x$ to $\ell_i$ we get
$$
\dist^2(x,\ell_i) 
=  \big(a_i^T(x - d_i)\big)^2
=  (a_i^Tx - b_i)^2
= (a_i^Tx)^2 + (b_i)^2 - 2\cdot a_i^Tx \cdot b_i 
,$$
where the scalar $b_i = a_i^Td_i$, for every $i \in [n]$.  The derivative of the squared distance from $x$ to $\ell_i$ is then

\[
\begin{split}
\frac{\partial}{\partial x} \dist^2(x,\ell_i) & =  \frac{\partial}{\partial x} \big((a_i^Tx)^2 + (b_i)^2 - 2\cdot a_i^Tx \cdot b_i \big) \\
& =  \frac{\partial}{\partial x} \big((a_i^Tx)^2 \big) 
+ \frac{\partial}{\partial x} (b_i)^2
-  \frac{\partial}{\partial x} \big( 2\cdot a_i^Tx \cdot b_i \big)  \\
& =  \frac{\partial}{\partial x} \big((a_i^Tx)^2 \big) 
-  \frac{\partial}{\partial x} \big( 2\cdot a_i^Tx \cdot b_i \big) \\
& =  2 (a_i^Tx) \cdot a_i -  2 b_i \cdot a_i ,\\
\end{split}
\]
since for every $a,b,c \in \REAL^d$ we have $(a^Tb)c=b(a^Tc)$, we get

\[
\frac{\partial}{\partial x} \dist^2(x,\ell) =   2 x \cdot (a_i^Ta_i) -  2 b_i \cdot a_i ,
\]
and summing it over every line $\ell_i \in L$ yields

\[ 
\begin{split}
 \frac{\partial}{\partial x}f(x) &  
= \frac{\partial}{\partial x} \sum_{i \in [n]} \dist^2(x,\ell_i) \\
& = \sum_{i \in [n]} \frac{\partial}{\partial x} \dist^2(x,\ell_i) \\
& =  \sum_{\ell \in [n]} \big(2 x \cdot (a_i^Ta_i) -  2 b_i \cdot a_i \big) \\
& = \sum_{i \in [n]} 2 x \cdot (a_i^Ta_i) - \sum_{i \in [n]} 2 b_i \cdot a_i.\\
\end{split}
\]
By $\frac{\partial}{\partial x}f(x) = 0$ we have

\[ 
\begin{split}
& \sum_{i \in [n]} 2 x \cdot (a_i^Ta_i) - \sum_{i \in [n]} 2 b_i \cdot a_i  = 0\\
& \sum_{i \in [n]} x \cdot (a_i^Ta_i)  = \sum_{i \in [n]} b_i \cdot a_i \\
&  x = \frac{\sum_{i \in [n]} b_i \cdot a_i}{\sum_{i \in [n]} (a_i^Ta_i)} \\
&  x = \frac{\sum_{i \in [n]} b_i \cdot a_i}{\sum_{i \in [n]} \norm{a_i}^2} .\\
\end{split}
\]
For every $i \in [n]$, substitute $b_i = a_i^Td_i$ and consider $(a_i^T d_i)a_i = \norm{a_i}^2 d_i$ finally yields

\[
\begin{split}
&  x = \frac{\sum_{i \in [n]} (a_i^Td_i) \cdot a_i}{\sum_{i \in [n]} \norm{a_i}^2} 
= \frac{\sum_{i \in [n]}  \norm{a_i}^2 \cdot d_i }{\sum_{i \in [n]} \norm{a_i}^2} \\
\end{split}
\]

\end{proof}
By the last theorem, given a set $L$ of lines in $\REAL^d$, one can computes its 1-line mean in linear time in the input size. Consider the following procedure: (1) randomly partition $L$ into $k$ clusters of lines, (2) compute each cluster's mean and (3) sum the squared distances from each line to its nearest mean in the cluster. Run the last procedure iteratively until a tunable stop condition is being achieved - wroks in practice and is measure in the next section - although, unlike our $k$-line median, its error and running time are unbounded.

\section{Algorithm~\ref{algorithm - a,b approx}: \bicriteria}

An $\alpha$-approximation for the $k$-median of a set $L$ of $n$ lines, is a set $P_\alpha \subseteq \REAL^d$ of $k$ points such that
$$
\cost(L, P_\alpha) \leq \alpha \cdot \min_{P^* \subseteq \REAL^d, \abs{P} = k} \cost(L, P^*)
.$$
A $\beta$-approximation for the $k$-median of $L$ is a set $P_\beta \subseteq \REAL^d$ of $\beta k$ points such that
$$
\cost(L, P_\beta) \in \min_{P^* \subseteq \REAL^d, \abs{P} = k} \cost(L, P^*)
.$$
An $(\alpha,\beta)$-approximation, also known as bi-criteria approximation, is a mixture of the two above approximations: it is a set $B \subseteq \REAL^d$ of $\beta k$ points such that
$$
\cost(L, B) \leq \alpha \cdot \min_{P^* \subseteq \REAL^d, \abs{P} = k} \cost(L, P^*)
.$$

Algorithm 2 gets a set $L$ on $n$ lines in $\REAL^d$, and outputs such an $(\alpha,\beta)$-approximation for $L$, where $\alpha=4\rho^2$ and $\beta \in O(dk \log (k) \log n)$, as proved in Theorem~\ref{theorem - the output of a,b approx is indeed a,b approx}.

Algorithm 2 is a special case of the framework for computing bi-criteria approximation as suggest Theorem 2.2 in~\cite{feldman2007bi}.

In order to apply this theorem for the $k$-median of lines problem, we need three ingredient:

1. A bound on the VC-dimension ("complexity") that corresponds to the $k$-median of lines problem, as formally stated in Definition~\ref{definition - k mean} and bounded in Corollary~\ref{corollary - the range space is O(dklog k)}. This VC-dimension $d^*$ determines the required size of the random sample that is picked in Line 4 of Algorithm~\ref{algorithm - a,b approx}, as stated in Corollary~\ref{corollary - (...)-median of S is a (...)-median of F}.

2. An algorithm that computes a provably robust median of the sample as in Definition~\ref{definition - robust median}. This is done via the \constantapprox~that was suggested on Algorithm~\ref{algorithm - 4 approx}. However, now the running time is faster due to the small sample size.

In Section~\ref{subsection - VC dim} we use classic results from the theory of PAC-learning to bound the VC-dimension. Then, in Section~\ref{section - bicriteria} we plug this bound with the \constantapprox~algorithm in the general bi-criteria framework to obtain the desired $(\alpha,\beta)$-approximation.
\subsection{Bound on the VC-Dimension} \label{subsection - VC dim}
We first define the classic notion of VC-dimension, that is usually related to the PAC-learning theory~\cite{LLS01}.

\begin{definition} [range space \cite{bigtotiny}]
A range space is a pair $(L, \ranges)$ where $L$ is a set, called ground set and $\ranges$ is a family (set) of subsets of $L$, called $\ranges$.
\end{definition}

\begin{definition} [VC-dimension \cite{bigtotiny}]
The VC-dimension of a range space $(L, \ranges)$ is the size
$\abs{S}$ of the largest subset $S \subseteq F$ such that
$$
\abs{\br{S \cap \range \mid \range \in \ranges}} = 2^{\abs{S}}.
$$
\end{definition}

\begin{definition} [range space of functions \cite{bigtotiny, harpeledsharir, newframework}]
Let $F$ be a finite set of functions from a set $\Q$ to $[0,\infty)$. For every $Q \in \Q$ and $r \geq 0$, let
$$
\range(F, Q, r) = \br{f \in F \mid f(Q) \geq r}
.$$
Let
$$
\ranges(F) = \br{\range(F, Q, r) \mid Q \in \Q, r \geq 0}.
$$
Finally, let $\R_{\Q,F} = \of{F, \ranges(F)}$ be the range space induced by $\Q$ and $F$.
\end{definition}
To bound the VC-dimension, we use the following theorem that bounds VC-dimension of range spaces that correspond to polynomial functions.

In what follows, $\sgn(x)$ denotes the sign of $x \in \REAL^d$. More precisely, $\sgn(x) = 1$ if $x >$ 0, $\sgn(x) = -1$ if $x < 0$, and $\sgn(x) = 0$ otherwise.

%theoram - Warren, the vc dim of n polinomials is small
\begin{theorem} [Theorem 3 in \cite{War} and Lemma 6 in \cite{bigtotiny}] \label{theoram - Warren, the vc dim of n polinomials is small}
Let $\br{f_1, \ldots, f_m}$ be real polynomials in $d^* < m$ variables, each of degree at most $b \geq 1$. Then the number of sign sequences $(\sgn\of{f_m(x)}, \ldots,  \sgn\of{f_1(x)})$, $x \in \REAL^d$, that consist of the terms $1,-1$ is at most $\of{\frac{4eb m}{d^*}}^{d^*}$.
\end{theorem}

%corollary - Warren, the number of distinct sequences is less then 2^m
\begin{corollary} [Corollary 3.1 in \cite{War}] \label{corollary - Warren, the number of distinct sequences is less then 2^m}
If $b \geq 2$ and $m \geq 8d^*\log_2 b$, then the number
of distinct sequences as in the above theorem is less than $2^m$.
\end{corollary}
Reformulating the query space for $k$-median of lines as a set of $n$ polynomial functions, allow us to bound the VC-dimension of the query space via Corollary~\ref{corollary - Warren, the number of distinct sequences is less then 2^m}. The proof of the following corollary is very similar to the proof in~\cite{bigtotiny} for the case of $k$ centers medians for given a finite set of $n$ lines.

%corollary - the range space is O(dklog k)
\begin{corollary} \label{corollary - the range space is O(dklog k)}
Let $\Q_k$ be the family (set) of all sets which are the union of $k$ points in $\REAL^d$ and let $L=\br{\ell_1, \ldots, \ell_n}$ be a set of $n$ lines in $\REAL^d$. 
Let $F^*_L = \br{f_1, \ldots, f_n}$, where $f_i(Q)= \dist(\ell_i, Q)$ for every $i \in [n]$ and $Q \in \Q_k$. Then the dimension of the range space $\R_{\Q_k,F^*_L}$ that is induced by $\Q_k$ and $F^*_L$ is $O(dk \log k)$.
\end{corollary}
\begin{proof}
We first show that for the case $k = 1$, the VC-dimension of the range space $\R_{\Q_k,F^*_L}$ is $O(d)$. Then the result follows from the fact that the $k$-fold intersection of range spaces of VC-dimension $O(d)$ has VC-dimension $O(dk \log k)$ \cite{BEHW89,EA07}.

If $n < d$ then the result is immediate. Thus, we consider the case $n > d$. We will first argue that the Euclidean distance function from a point to a line can be written as a polynomial in $O(d)$ variables. Indeed, let $Q \in Q_k$. For every $i \in [n]$, we define $\dist(\ell_i,Q) = \min_{q \in Q}\norm{(q-t_i)^TX_i}$, where $X_i \in \REAL^{d \times (d-1)}$ s.t. $X_i^TX_i=I$, $t_i \in \REAL^d$ correspond to the line $\ell_i$ and $q \in \REAL^d$ is the unknown point. Therefore, for some $r>0$, $\dist^2(\ell_i, Q) - r^2 = \min_{q \in Q} \norm{(q-t_i)^TX_i}^2-r^2$ is a polynomial of constant degree $b=2$ with $d^* \in O(d)$ variables. 

Consider a subset $Y \subset F$ s.t. $\abs{Y} = m = 8d^*$, and denote the functions in $Y$ by $\br{f_1, \ldots, f_m}$. Our next step is to upper bound the number of distinct ranges in our range space $\R_{\Q_k,F^*_L}$ that intersects $Y$, for $k = 1$. Let $Q \in \Q_k$ and $r \geq 0$. By defining
$$
\range(L, Q, r) = \br{\ell \in L \mid \dist(\ell,Q) \leq r}
$$
for every $i \in [n]$ we have $\dist(\ell_i,Q) \geq r$, if and only if $\sgn\of{\min_{q \in Q}\norm{(q-t_i)^TX_i}^2-r^2} \geq 0$, where $\sgn(x) = 1$ if $x >$ 0, $\sgn(x) = -1$ if $x < 0$, and $\sgn(x) = 0$ otherwise. 
\\\\

Thus, the number of ranges is at most
$$
\abs{\ranges(L)} = 
$$
$$
\abs{\br{\sgn\of{\min_{q \in Q}\norm{(q-t_1)^TX_m}^2-r^2}, \ldots, \sgn\of{\min_{q \in Q}\norm{(q-t_m)^TX_m}^2-r^2} \mid Q \in \Q_k, r \geq 0}}
.$$

We also observe that for every sign sequence that has zeros, there is a sign sequence corresponding to the same range that only contains $1$ and $-1$ (this can be obtain by infinitesimally changing r). Thus, substituting $f_i = \dist(\ell_i,Q)$ for every $i \in [n]$ in Theorem \ref{theoram - Warren, the vc dim of n polinomials is small}, yields that the number of such sequences is bounded by $\of{\frac{4e b m}{d^*}}^{d^*}$, where $b = 2$ and is the above polynomials degree. By Corollary \ref{corollary - Warren, the number of distinct sequences is less then 2^m}, since $b = 2$ and $m = 8d^*$, the number of such ranges is less than $2^m$. At the same time, a range space with VC-dimension $d$ must
contain a subset $Y$ of size $d$ such that any subset of $Y$ can be written as $Y \cap \range$ for some $\range \in \ranges(F)$, which implies that the number of such sets is $2^d$. Since this is not possible for $Y$ if $m \geq 8d^*$, we know that the VC dimension of our range space is bounded by
$8 d^* \in O(d)$ (for the case $k = 1$). 

Now the result follows by observing that, in the case of $k$ centers, every range is obtained by taking the intersection of $k$ ranges of the range space for $k = 1$. Hence, by Corollary~\ref{corollary - the range space is O(dklog k)} the dimension of the range space $\R_{\Q_k,F^*_L}$ that is induced by $\Q_k$ and $F^*_L$ is $O(dk \log k)$.
\end{proof}

\subsection{Analysis of Algorithm 2: \bicriteria} \label{section - bicriteria}

Intuitively, a coreset for $k$-median of lines supposes to contain the ``important" input lines. This is formalized by the sensitivity sampling approach in Theorem~\ref{theorem - total sensitivity can be computed given a,b-approx}. Not surprisingly, as in many other coreset constructions, this sampling distribution is based on the optimal solution of the query space (i.e., the query that minimizes the cost). However, computing the optimal query, or $k$-median in our case, is the original motivation for construction the coreset. To find a lee-way from this chiecken-and-egg problem, we suggest a very rough approximation for the optimal solution, called $(\alpha,\beta)$-approximation. Algorithm~\ref{algorithm - a,b approx} computes it efficiently, and this type of approximation suffices to bound the sensitivities as is shown in Section~\ref{section - sensitivity}.

\begin{definition} [$\alpha,\beta$-approximation] \label{definition - a,b approx}
Let $L$ be a finite set of lines in $\REAL^d$, $k \geq 1$ be an integer and $P^* \subseteq \REAL^d$ be a $k$-median of $L$; See Definition \ref{definition - k mean}. Then, for every $\alpha,\beta \geq 0$, a set $B \subseteq \REAL^d$ of $k\beta$ points is called $(\alpha,\beta)$-\emph{approximation} of $L$, if
$$
\cost(L, B) \leq \alpha \cdot \cost(L, P^*).
$$
If $\beta=1$ then $B$ is called an $\alpha$-\emph{approximation} of $L$. , if $\alpha=\beta=1$ then $B$ is called \emph{the optimal solution}.
\end{definition}

To prove that the output of Alg.~\ref{algorithm - a,b approx} is indeed an $(\alpha,\beta)$-approximation for small values of $\alpha$ and $\beta$, we use a generic framework by Feldman and Langberg \cite{newframework}.

The first definition is for robust median. That is a point that aims to minimizes the sum of distances from an input set of lines, up to a small fraction $\gamma$ of outliers that can be ignored. Our approximation can serves only fraction of $(1-\eps)\gamma$ of its closest input lines, instead of the desired $\gamma$. Finally, we aim to get only $(\alpha,\beta)$-approx for this robust median. Hence, the following approximation can be regarded as a triple-criteria $(\alpha,\beta,\eps)$ for the optimal median that ignores a $\gamma$ fraction of outliers.

\begin{definition} [robust median \cite{newframework}] \label{definition - robust median}
Let $F$ be a set of $n$ functions from a set $X$ to $[0,\infty)$. Let $0 < \eps, \gamma < 1$, and $\alpha > 0$. For every $x \in X$, let $F_x$ denote the $\ceil[\big]{\gamma n}$ functions $f \in F$ with the smallest value $f(x)$. Let $Y \subseteq X$,
and let $G$ be the set of the $\ceil[\big]{(1-\eps) \gamma n}$ functions $f \in F$ with smallest value $f(Y) = \min_{y\in Y} f(y)$. The set $Y$ is called a $(\gamma, \eps, \alpha, \beta)$-median of $F$, if $\abs{Y} = \beta$ and
$$
\cost(G,Y) \leq \alpha \min_{x \in X} \cost(F_x,x),
$$
where $\cost : P(X) \times P(X) \to [0,\infty)$ and $P(X)$ is the power set of $X$.
\end{definition}

\begin{corollary} [\cite{newframework}] \label{corollary - (...)-median of S is a (...)-median of F}
Let $\eps \in (0, 1/10)$ and $\delta, \gamma \in (0, 1]$. Let $F$ be a set of $n \geq 1/(\gamma\eps)$ functions from $X$ to $[0,\infty)$, where the VC-dimension of the range space that induced by $F$ and $X$ is $d^*$. Suppose that we have an algorithm that receives any set $S \subseteq F$ of size
$$
\abs{S} \in \Theta \of{\frac{d^*+\log(1/\delta)}{\gamma^2\eps^4}},
$$
and returns a set $G , \abs{G} \leq \beta$ that contains a $((1 - \eps)\gamma, \eps, \alpha,1)$-median of $S$ in time \\$O(1) \cdot \mathbf{SlowMedian}$. Then a $(\gamma, 4\eps, \alpha, \beta)$-median of $F$ can be computed, with probability at least $1-\delta$, in time 
$$
O(1) \cdot \mathbf{SlowMedian}+ O\of{\abs{S}}
.$$
\end{corollary}

The following theorem forges a link between robust medians and $(\alpha,\beta)$-approximations.
\begin{theorem} [Generic Bi-criteria \cite{newframework}] [\label{theorem - new framework, output of a,b approx is bi criteria}
Let $F$ be a set of $n$ functions from a set $X$ to $[0,\infty)$. Let $\eps > 0$ and $\delta, \gamma < 1$. Let $\alpha, \beta \geq 0$. Then a set
$Z \subseteq X$ of size $\abs{Z} \leq \beta \log_2 n$ can be computed such that, with probability at least $1-\delta$,
$$
\cost(F, Z) \leq (1 + \eps)\alpha \cdot \min_{x\in X} \cost(F, x) .
$$
This takes time
$$
\mathbf{Bicriteria} = O(1) \cdot (nt + \log (n) \cdot \mathbf{SlowMedian} + \mathbf{SlowEpsApprox}),
$$
where:
\begin{itemize}
\item $t$ is an upper bound on the time it takes to compute $f(Y)$ for a pair $f \in F$ and $Y \subseteq X$ such that $\abs{Y} \leq \beta$.

\item $O(\mathbf{SlowMedian})$ is the time it takes to compute, with probability at least $1-\delta/2$, a $(3/4, \eps, \alpha, \beta)$-median for a set $F' \subseteq F$.

\item O($\mathbf{SlowEpsApprox})$ is the time it takes to compute a $(1, 0, \alpha, \beta)$-median for a set $F' \subseteq F$ of size
$\abs{F'} = O(1/\eps)$.

\end{itemize}
\end{theorem}

Alg.~\ref{algorithm - a,b approx} is the specification of the algorithm \textsc{Bicriteria} in \cite{newframework} for the case that $F$ is a set $L$ of lines, $\beta$ is a positive integer, $k \geq 1$, $\eps=1/11$ and $\alpha=4$. By plugging our \constantapprox~algorithm from Lemma~\ref{lemma - 4 approx} and the bound on the VC-dimension in Theorem~\ref{theoram - Warren, the vc dim of n polinomials is small} in the bicriteria framework from~\cite{newframework}, we conclude that Alg~\ref{algorithm - a,b approx} returns an $(\alpha,\beta)$-approx as follows.

%theorem - the output of a,b approx is indeed a,b approx 
\begin{theorem} \label{theorem - the output of a,b approx is indeed a,b approx}
Let $L$ be a set of $n$ lines in $\REAL^d$, $k\geq1$ be an integer, $\delta \in (0,1)$ and 
$$
m \geq c\of{dk \log_2 k +\log_2\of{\frac{1}{\delta}}}
,$$ 
for a sufficiently large constant $c>1$ that can be determined from the proof. Let $B$ be the output set of a call to $\bicriteria(L,m)$; See Algorithm \ref{algorithm - a,b approx}. Then, 
\begin{equation} \label{equation - size of ab approx}
\abs{B} \in O\of{\log n \of{dk \log k+\log(1/\delta)}}
\end{equation}
and with probability at least $1-\delta$,
$$
\cost(L,B) \leq 4 \rho^2 \cdot \min_{P \subseteq \REAL^d, \abs{P}=k}\cost(L,P)
.$$
Moreover, $B$ can by computed in $O\of{nd^2k\log k +m^2 \log n}$ time.
\end{theorem}

\begin{proof}
Let $i \in [\lceil \log_2n \rceil]$, and consider the values of $X, X', S, B$ and $G$ during the execution of the $i$th iteration of the main ``while'' loop in Line 3 of Alg.~\ref{algorithm - a,b approx} . That is, identify $X$ as the set of lines that were computed in Line 8 during the execution of the $(i-1)$th iteration, $S$ be a set of at least $m$ lines that was randomly chosen from $X$ in Line 4, and $G \subseteq \REAL^d$ be the centroid set that is computed in Line 5 during the $i$th iteration.

Let $p^*\in\REAL^d$. Substituting $P=\br{p^*}, L=S, m= \lceil10\abs{S}/11 \rceil$ and $k=1$ in Theorem~\ref{theorem - 4 approx is robust} yields that there is a set $P''=\br{p'}\subseteq G$ such that
\begin{equation} \label{equation - G contains (...)-median of single point} 
\cost\of{\closest\of{S,\br{p'},\frac{10\abs{S}}{11}}, \br{p'}} \leq 4 \rho^2 \cdot \cost\of{\closest\of{S,\br{p^*},\frac{10\abs{S}}{11}}, \br{p^*}}.
\end{equation}
Hence, there is a point $p' \in G$ which is a $((1-\eps)\gamma, \eps, \alpha, 1)$-median of $S$, where $\gamma=1,\eps=1/11$ and $\alpha=4 \rho^2$, that can be found via exhaustive search over every point in $G$, in $O\of{\abs{G} \cdot \abs{S}}$ time.

Next, recall the definition of $\Q_k=\br{Q\subseteq\REAL^d\mid |Q|=k}$ as the family of all sets which are the union of $k$ points in $\REAL^d$. For every $Q \in \Q_k$, let $f_j(Q)=\dist(\ell_j,Q)$ and $F^*_L$ be the union of these functions, as defined in Corollary~\ref{corollary - the range space is O(dklog k)}. We get that the VC-dimension of the range space that is induced by $\Q_k$ and $F$ is $O(dk \log k)$. 

Consider that
$$
\abs{S} \in \Theta \of{dk \log_2 k +\log_2\of{\frac{1}{\delta}}}
,$$ 
substituting in Corollary~\ref{corollary - (...)-median of S is a (...)-median of F}, $d^* = dk \log k,\gamma=1,\eps=1/11, \alpha = 4 \rho^2$ and $\beta = O(m^2)$, together with \eqref{equation - G contains (...)-median of single point} yields that with probability $1-\delta$, $G$ is a $(1,4/11,4\rho^2,O(m^2))$-median of $X$.

\paragraph{Running time.}
Note that:
\begin{itemize}
\item The time it takes to compute $f(Y)=\dist(\ell,Q)$, i.e., the Euclidean distance from $\ell$ to $Q$ for a single line $\ell \in X$, and a set $Q \subseteq \REAL^d$ of $\abs{Q}=O\of{dk \log k}$ points is $t = O\of{d^2k \log k}$.

\item Defining $\mathbf{SlowMedian}$ to be the time $O(d^2m^2)$ that it takes to compute $G$ by a call to $\constantapprox(S)$, and $\mathbf{SlowEpsApprox}$ is the time it takes to compute a $(1,0,\alpha)$-approximation for a set of $1/\eps$ lines ($\alpha$ and $\eps$ are defined as above), which, in turn, takes $O(d^2)$ time; See Lemma~\ref{lemma - 4 approx}.
\end{itemize}

Using this and substituting $\alpha = 4 \rho^2, \beta = O(m^2), X=\Q_k$ in Theorem~\ref{theorem - new framework, output of a,b approx is bi criteria}, when $F, \eps$ and $\gamma$ as above, yields that, with probability at least $1-\delta$,
$$
\cost(L,B) \in O(1) \cdot \min_{P \subseteq \REAL^d, \abs{P}=k}\cost(L,P),
$$ 
and the running time it takes to compute $B$ is $O\of{nd^2k\log k +m^2 \log n}$. This proves the theorem~\ref{theorem - the output of a,b approx is indeed a,b approx}.

\end{proof}

\section{Algorithm~\ref{algorithm - sensitivityoftranslatedlines}: \sensitivityoftranslatedlines} \label{section - sensitivity}

The generic coreset construction that we use in Algorithm~\ref{algorithm - coreset} is essentially a non-uniform sample from a distribution that is based on the~\abapprox~from Section~\ref{section - bicriteria}. This distribution is known as sensitivity which we define and bound in this section.

%Example of lines sets $L$ and $L'$ - figure
\begin{figure}[h]
		\caption{\textbf{Illustration of Lemma~\ref{theorem - total sensitivity can be computed given a,b-approx} in the planar case.} \textbf{(left)} A set $L=\br{\ell_1, \ldots, \ell_{10}}$ and its $(\alpha,\beta)$-approximation $\hat{P}=\br{\hat{p_{1}},\hat{p_{2}},\hat{p_{3}},\hat{p_{4}}}$ for its $k$-median where $k=2, \alpha=4$, and $\beta=2$. \textbf{(right)} Each line in $L$ is translate to its nearest points in $\hat{P}$, resulting in the set $L'=\br{\ell_1', \ldots, \ell_{10}'}$.}
		\includegraphics[height=7cm, width=16cm]{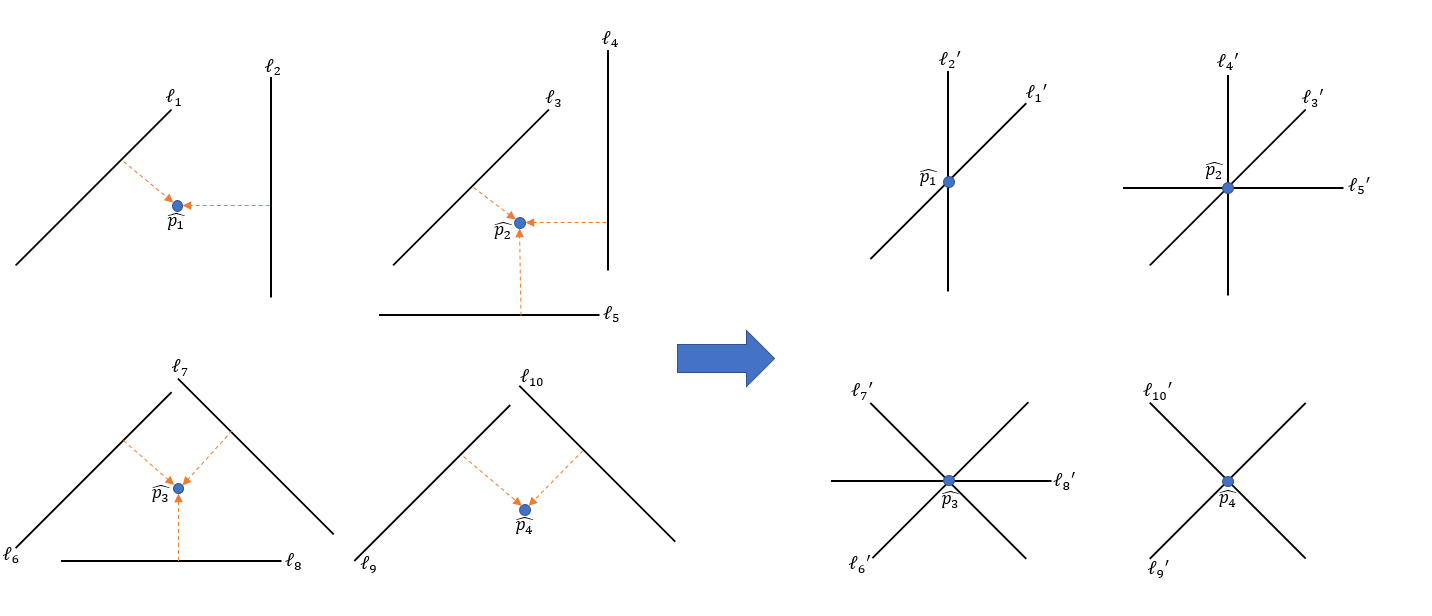}
		\label{figure - L to L'}
\end{figure}

The following definition is a special case of the definition of sensitivity in~\cite{newframework} for our problem of $k$-median of lines.
%Sensitivity - definition
\begin{definition} [sensitivity of lines] \label{definition - sensitivity}
Let $L$ be a set of $n$ lines in $\REAL^d$ and $k\geq1$ be an integer. We define the \emph{sensitivity} of a line $\ell \in L$ by
\begin{align} \label{single line sensitivity}
S_{L,k}^*(\ell) 
= \sup_{P \subseteq \REAL^d, \abs{P}=k} \frac{\dist(\ell,P)}{\cost(L,P)}
= \sup_{P \subseteq \REAL^d, \abs{P}=k} \frac{\dist(\ell,P)}{\sum_{\ell' \in L} \dist(\ell', P)},
\end{align}
where the $\sup$ is over every set of $k$ points in $\REAL^d$ such that the denominator is positive. The \emph{total sensitivity} is defined to be the sum over these sensitivities, $S_k^*(L) = \sum_{\ell \in L} S_{L,k}^*(\ell)$. The function $S_{L,k} : L \to [0, \infty)$  is a \emph{sensitivity bound} for $S_{L,k}^*$, if for every $\ell \in L$ we have $S_{L,k}(\ell) \geq S_{L,k}^*(\ell)$. The \emph{total sensitivity bound} is then defined to be
\begin{align} \label{total lines sensitivity}
S_k(L)= \sum_{\ell \in L} S_{L,k}(\ell).
\end{align}
\end{definition}

The sensitivity bound that is computed in Algorithm~\ref{algorithm - sensitivityoftranslatedlines} is splitted into two parts. The first part is related to the distance from each line to its bicriteria approx $\hat{P}$, ignoring its slope. The second part depends only on the slope of its line. In the second case the direction of the line is translated into a unit vector $p$, and the $i$th query point is translated into a unit vector $c_i$ with a corresponding weight $w_i$. The sensitivity of this new query space was bounded by Feldman and Schulman~\cite{outliers-resistance} as stated in the following Theorem.

\begin{theorem} [\cite{outliers-resistance}]\label{theorem - total sensitivity of k weighted points theorem}
Let $P \subseteq \REAL^d$ be as set of $n$ points in $\REAL^d$, $k \geq 1$ be an integer and $\delta \in (0,1)$. For every $p \in P$, let
$$
s^*(p):=\sup_{\substack{\br{c_1, \ldots, c_k} \subseteq \REAL^d\\ \br{w_1, \ldots, w_k} \in [0,\infty)^k}}   \frac{\min_{i \in [k]} w_i\cdot \dist\of{p,c_i}}{\sum_{p' \in P} \min_{i \in [k]} w_i\cdot \dist\of{p',c_i}}
,$$
where the $\sup$ is over every $k$ points $c_1,\ldots,c_k$ with corresponding weights $w_1,\ldots,w_k$ such that the denominator is positive. Then, with a probability at least $1-\delta$, a function $s:P\to  (0, \infty)$ such that $s(p)\geq s^*(p)$ and 
\begin{equation} \label{equation - total sensitivity of k weighted points theorem}
\sum_{p \in P} s(p) \in k^{O(k)}\log n,
\end{equation} 
can be computed in $ndk^{O(k)}$ time.
\end{theorem}

The following lemma reduces the sensitivity bound of each line to the sensitivity bound on its translated line in $L'$ as in Fig.~\ref{figure - L to L'} plus an additive term that is the distance of the line to the bicriteria apporximatoin $\hat{P}$.

%Upper bound for S(L) - theoram
\begin{lemma} \label{theorem - total sensitivity can be computed given a,b-approx}
Let $L = \br{\ell_1, \ldots, \ell_n}$ be a set of $n$ lines in $\REAL^d$, $k\geq 1$ be an integer, $\alpha, \beta > 0, \delta \in (0,1)$ and let $\hat{P}  \subseteq \REAL^d$ be an \abapprox~for the $k$-median of $L$. For every $\ell \in L$, let $\ell'$ be the line that is parallel to $\ell$ and passes through the closest point to $\ell$ in $\hat{P}$, and $L' = \br{\ell'\mid \ell\in L}$ denote their union; See Fig.~\ref{figure - L to L'}.

Then, given $\hat{P}$, with a probability at least $1-\delta$ a sensitivity bound $S_{L,k}(\ell)$ for $S^*_{L,k}(\ell)$ can be computed in time
$$
O\of{d^2 n\log (n) k \log k } + ndk^{O(k)}
,$$
and its total sensitivity bound is $\alpha\beta k^{O(k)}\log n$.

\end{lemma}

\begin{proof}
The proof splits into two parts, as follows: in the first part we bound $S^*_k(L)$ via $S_k^*(L')$, and in the second part we bound $S_k^*(L')$ with $S_k(L')$.

\paragraph{Bound on $S^*_k(L)$.}
Let $P \subseteq \REAL^d$ be a set of $k$ points, $\ell \in L$ and $\ell'$ be its corresponding line in $L'$. Let $p' \in P$ denote the closest point to $\ell'$ in $P$. Ties broken arbitrarily). We then have,
\begin{align}
\dist(\ell,P) 
& \leq \dist(\ell,p')  \label{equation - sensitivity_triangle_inequality - b} \\
& \leq \rho \of{D\of{\ell,\pi\of{p',\ell'}}+D\of{\pi\of{p',\ell'},p'}}   \label{equation - sensitivity_triangle_inequality - c} \\
& = \rho \of{D\of{\ell,\pi\of{p',\ell'}}+D\of{\ell', p'}} \nonumber \\
& = \rho \of{D\of{\ell,\ell'}+D\of{\ell',P}},   \label{equation - sensitivity_triangle_inequality - d}
\end{align}

where~\eqref{equation - sensitivity_triangle_inequality - b} is by the definition of $p'$, and \eqref{equation - sensitivity_triangle_inequality - c} holds by the triangle inequality. Finally, \eqref{equation - sensitivity_triangle_inequality - d} holds since $\dist(\ell,q)=\dist(\ell,\ell')$ for every $q\in \ell'$, as $\ell$ and $\ell'$ are parallel.

Next, we bound the sum of distances from the set $L'$ of projected lines to $P$ by
\begin{align}
\cost(L',P)  
& = \sum_{\ell'\in L'}\dist(\ell',P) \nonumber \\
&\leq \rho \sum_{\ell'\in L'}\big(\dist(\ell',\ell)+\dist(\ell,P)\big)
 \label{equation - sensitivity_cost(L',hatP)_leq_2alphacost(L,P) - a} \\
& =\rho \of{ \cost(L,\hat{P})+\cost(L,P)}  \nonumber \\
& \leq \rho (\alpha+1) \cdot \cost(L,P), \label{equation - sensitivity_cost(L',hatP)_leq_2alphacost(L,P) - b}
\end{align}
where~\eqref{equation - sensitivity_cost(L',hatP)_leq_2alphacost(L,P) - a} holds as in~\eqref{equation - sensitivity_triangle_inequality - d} and~\eqref{equation - sensitivity_cost(L',hatP)_leq_2alphacost(L,P) - b} since $\hat{P}$ in an \abapprox~for $L$. We then obtain a bound on the rightmost term in \eqref{equation - sensitivity_triangle_inequality - d}, as

\begin{align}
\dist(\ell',P)  
& = \frac{\dist(\ell',P)}{\cost(L',P)} \cdot \cost(L',P) \nonumber \\
&\leq S_{L',k}^*(\ell') \cdot \cost(L',P)   \label{equation - sensitivity_bound_on_dist(ell',P) - a} \\
& \leq  S_{L',k}^*(\ell') \cdot \rho (\alpha+1) \cdot \cost(L,P),  \label{equation - sensitivity_bound_on_dist(ell',P) - b} 
\end{align}
where~\eqref{equation - sensitivity_bound_on_dist(ell',P) - a} holds by the definition of $S_{L',k}^*(\ell')$ and~\eqref{equation - sensitivity_bound_on_dist(ell',P) - b} by~\eqref{equation - sensitivity_cost(L',hatP)_leq_2alphacost(L,P) - b}.

We now bound the term in the left hand side of \eqref{equation - sensitivity_triangle_inequality - d} by
\begin{align}
\dist(\ell,\ell')  
& = \frac{\dist(\ell,\ell')}{\cost(L,P)} \cdot \cost(L,P) \nonumber \\
& \leq \frac{\alpha\cdot\dist(\ell,\ell')}{\cost(L,\hat{P})} \cdot \cost(L,P)   \label{equation - sensitivity_bound_on_dist(ell,ell') - a} \\
& = \frac{\alpha\cdot\dist(\ell,\hat{P})}{\cost(L,\hat{P})} \cdot \cost(L,P),  \label{equation - sensitivity_bound_on_dist(ell,ell') - b} 
\end{align}
where~\eqref{equation - sensitivity_bound_on_dist(ell,ell') - a} holds by the definition of $\hat{P}$ as an $(\alpha,\beta)$-approximation of $L$, and~\eqref{equation - sensitivity_bound_on_dist(ell,ell') - b} since $\ell$ and $\ell'$ are parallel. Plugging \eqref{equation - sensitivity_bound_on_dist(ell',P) - b} and \eqref{equation - sensitivity_bound_on_dist(ell,ell') - b}  in \eqref{equation - sensitivity_triangle_inequality - d} yields

\begin{align} 
\dist(\ell,P) 
& \leq  \rho \big( \dist(\ell,\ell') +  \dist(\ell',P) \big) \nonumber \\
& = \rho \cdot \cost(L,P) \cdot \of{\frac{\alpha\cdot\dist(\ell,\hat{P})}{\cost(L,\hat{P})} + \rho (\alpha+1) \cdot S_{L',k}^*(\ell')}. \label{equation - dist(l,p) upper bound}
\end{align}

Let $ \displaystyle \ P^* \in \argsup _{P' \subseteq \REAL^d, \abs{P'} = k} \frac{\dist(\ell,P')}{\cost(L,P')}$ denote a set of points that maximizes the sensitivity of $\ell$ s.t. the denominator is positive. We then obtain that a sensitivity of $\ell$ is

\begin{align}
S_{L,k}^*(\ell)   
& = \frac{\dist(\ell,P^*)}{\cost(L,P^*)} \nonumber \\
& \leq \frac{\rho (\dist(\ell,\ell') +  \dist(\ell',P^*))}{\cost(L,P^*)} \label{equation - S^*_L is bounded by S^*_L_tag - a} \\
& \leq \frac{\rho \cdot \cost(L,P^*) \cdot \of{\frac{\alpha\cdot\dist(\ell,\hat{P})}{\cost(L,\hat{P})} + \rho(\alpha+1) \cdot S_{L',k}^*(\ell')}}{\cost(L,P^*)} \label{equation - S^*_L is bounded by S^*_L_tag - b} \\
& = \frac{\rho \alpha \dist(\ell,\hat{P})}{\cost(L,\hat{P})} +  \rho^2(\alpha+1)\cdot S_{L',k}^*(\ell'), \nonumber
\end{align}
where \eqref{equation - S^*_L is bounded by S^*_L_tag - a} is by~\eqref{equation - sensitivity_triangle_inequality - d} and \eqref{equation - S^*_L is bounded by S^*_L_tag - b} holds by substituting $P=P^*$ in~\eqref{equation - dist(l,p) upper bound}.
Summing the last inequality over every $\ell \in L$ yields
\begin{align}
S^*_k(L) 
& = \sum_{\ell \in L} S^*_{L,k}(\ell) \nonumber \\
& \leq \sum_{\ell \in L} \of{ \frac{\rho\alpha\cdot\dist(\ell,\hat{P})}{\cost(L,\hat{P})} + \rho^2(\alpha+1) \cdot S_{L',k}^*(\ell')} \nonumber \\
& = \rho\alpha + \rho^2(\alpha+1) \cdot S_k^*(L').\label{equation - original lines sensitivity upper bound}
\end{align}

\paragraph{Bound on $S_k^*(L')$.} 
Let $\hat{p}_i \in \hat{P}$ denote the $i$th point of $\hat{P}$ and $L'_i \subseteq L'$ denote the subset of lines in $L'$ that intersect at $\hat{p_i}$, for every $i \in [\lceil k\beta\rceil]$. 
The total sensitivity is then bounded by
\begin{equation} \label{total lines L' sensitivity}
\begin{split}
 S_k^*(L') 
& = \sum_{\ell' \in L'} \sup_{P \subseteq \REAL^d, \abs{P} = k} \frac{\dist(\ell',P)}{\cost(L',P)} 
= \sum_{i=1}^{k\beta} \sum_{\ell' \in L'_{i}} \sup_{P \subseteq \REAL^d, \abs{P} = k} \frac{\dist(\ell',P)}{\cost(L',P)}\\
& \leq  \sum_{i=1}^{k\beta} \sum_{\ell' \in L'_{i}} \sup_{P \subseteq \REAL^d, \abs{P} = k} \frac{\dist(\ell',P)}{\cost(L'_{i},P)} 
= \sum_{i=1}^{k\beta} S_k^*(L'_{i})
,\\
\end{split}
\end{equation}
where the sup is over only positive values of the denominator. Let $i \in [\lceil k\beta\rceil]$. Without loss of generality, assume that $\hat{p_i}$ is the origin of $\REAL^d$, otherwise, translate the coordinate system. By Thales's theorem, for every $\ell' \in L_i'$ and $p \in \REAL^d$ we have
\begin{equation} \label{equation - Thales' theorem}
\dist(\ell',p) = \norm{p} \cdot \dist\of{\ell',\frac{p}{\norm{p}}}.
\end{equation}

Recall $\proj(X,Y) \in \arginf_{y \in Y} \dist(X,y)$ to be the closest point in $Y$ to $X$, for every two sets $X,Y \subseteq \REAL^d$, ties are broken arbitrary. The total sensitivity is then bounded by
\begin{equation} \label{equation - sensitivity is the proj}
\begin{split} 
S_k^*(L'_{i}) 
& = \sum_{\ell' \in L'_{i}} \sup_{P \subseteq \REAL^d, \abs{P} = k} \frac{\dist(\ell',P)}{\cost(L'_{i},P)} \\
& = \sum_{\ell' \in L'_{i}} \sup_{P \subseteq \REAL^d, \abs{P} = k} \frac{\dist(\ell',\proj(\ell',P))}{\sum_{\ell \in L_i'} \dist(\ell,\proj(\ell,P))} \\
& = \sum_{\ell' \in L'_{i}} \sup_{P \subseteq \REAL^d, \abs{P} = k}
\frac{\norm{\proj(\ell',P)} \cdot \dist\of{\ell',\frac{\proj(\ell',P)}{\norm{\proj(\ell',P)}}}}
{\sum_{\ell \in L_i'} \norm{\proj(\ell,P)} \cdot \dist\of{\ell,\frac{\proj(\ell,P)}{\norm{\proj(\ell,P)}}}},
\end{split}
\end{equation}
where the denominator is positive and the last equality is by~\eqref{equation - Thales' theorem}. For every $\ell' \in L'_i$ and  
$$
Q \in \argsup_{P \subseteq \REAL^d, \abs{P} = k}
\frac{\norm{\proj(\ell',P)} \cdot \dist\of{\ell',\frac{\proj(\ell',P)}{\norm{\proj(\ell',P)}}}}
{\sum_{\ell \in L_i'} \norm{\proj(\ell,P)} \cdot \dist\of{\ell,\frac{\proj(\ell,P)}{\norm{\proj(\ell,P)}}}}
,$$
let $u(\ell') = \norm{\proj(\ell',Q)}$. Let $\mathbb{S}^{d-1}=\br{q \in \REAL^d \mid \norm{q}=1}$ denote the unit sphere of $\REAL^d$. We then get
\begin{equation} \label{equation - sensitivity of k weighted points to lines}
S_k^*(L'_{i}) = \sum_{\ell' \in L'_{i}} \sup_{P \subseteq \mathbb{S}^{d-1}, \abs{P} = k} \frac{u(\ell') \cdot \dist(\ell',\proj(\ell',P))}{\sum_{\ell \in L_i'} u(\ell) \cdot \dist(\ell,\proj(\ell,P))},
\end{equation} 
and note that now the $\sup$ is over every $k$ points on the unit sphere such that the denominator is positive. For every set $B \subseteq \mathbb{S}^{d-1}$ of $k$ points, let $-B=\br{-b \mid b \in B}$. For a line $\ell$ that intersects the origin, let $q_{\ell}\in \mathbb{S}^{d-1} \cap \ell$. 
The distance from $\ell$ to its closest unit vector $b\in B$ is then
\begin{equation} \label{equation - distance to set from line is the same as the angle}
\dist(\ell,B)=\dist(\ell,b)
\in \Theta(1) \cdot \dist(q_{\ell},\br{b, {-b}})
=\Theta(1) \cdot \dist(q_{\ell},B\cup -B),
\end{equation}
where in the second derivation we used the fact that the angle $\gamma\in[0,\pi/2)$ between two unit vectors $q_{\ell}$ and $b$ around the origin is the same as the distance between them up to a multiplicative constant factor.

By the combination of \eqref{equation - sensitivity of k weighted points to lines} and~\eqref{equation - distance to set from line is the same as the angle} we obtain
\begin{equation}
S_k^*(L'_{i})  \in \Theta(1) \cdot \sum_{\ell' \in L'_{i}}  \sup_{P \subseteq \mathbb{S}^{d-1}, \abs{P} = k} \frac{u(\ell') \cdot \dist(q_{\ell'},P\cup-P)}{\sum_{\ell \in L_i'} u(\ell) \cdot \dist(q_{\ell},P\cup-P)}. \label{equation - max over k points with mins is like max over 2k points - b}
\end{equation}
The distance $\dist$ in~\eqref{equation - max over k points with mins is like max over 2k points - b} gets as an input a set of $2k$ points, hence, the total sensitivity is 

\begin{equation} \label{equation - sensitivity of k less then 2k}
S_k^*(L'_{i})  
=\sum_{\ell' \in L'_i} S^*_{L'_i,k}(\ell')
\in \Theta(1) \cdot \sum_{\ell' \in L'_{i}} \sup_{Q \subseteq \REAL^d, \abs{Q}=2k} \frac{u(\ell') \cdot \dist(q_{\ell'},Q)}{\sum_{\ell \in L_i'} u(\ell) \cdot \dist(q_{\ell},Q)},
\end{equation}

Using \eqref{equation - sensitivity of k less then 2k} and plugging $P = \bigcup_{\ell' \in L'_i} q_{\ell'}$, $\delta$,
$$
s(q_{\ell'}) = \sup_{Q \subseteq \REAL^d, \abs{Q}=2k} \frac{u(\ell') \cdot \dist(q_{\ell'},Q)}{\sum_{\ell \in L_i'} u(\ell) \cdot \dist(q_{\ell},Q)},
$$
and $2k$ instead of $k$ in Theorem~\ref{theorem - total sensitivity of k weighted points theorem} yields that a sensitivity bound $s = S_{L',k}(\ell') \geq S_{L',k}^*(\ell')$ of each $\ell' \in L'$ can be computed in time  $ndk^{O(k)}$. The total sensitivity of $L'_i$ is then bounded by
\begin{equation} \label{equation - a sensitivity bound of L tag i can be computed}
S_k(L'_{i}) \in k^{O(k)}\log n.
\end{equation}
Summing this bound over every $i \in [\lceil k\beta \rceil ]$ yields
\begin{equation} \label{equation - total sensitiviti of Litag bounded with beta}
 S_k^*(L') \leq
 \sum_{i=1}^{k\beta} S_k(L'_{i})
 \in \beta k^{O(k)}\log n,
\end{equation}
where the first bound is by \eqref{total lines L' sensitivity} and the second from \eqref{equation - a sensitivity bound of L tag i can be computed}. Finally, we get
\begin{align}
S_k^*(L) 
& \leq \rho \alpha + \rho^2(\alpha+1) \cdot S_k^*(L') \label{equation - the total sensitivity is bounded by f(alpha,beta,k ^ O(k)) - a} \\
& \in \alpha \beta k^{O(k)}\log n, \label{equation - the total sensitivity is bounded by f(alpha,beta,k ^ O(k)) - b}
\end{align}
where \eqref{equation - the total sensitivity is bounded by f(alpha,beta,k ^ O(k)) - a} is by \eqref{equation - original lines sensitivity upper bound} and \eqref{equation - the total sensitivity is bounded by f(alpha,beta,k ^ O(k)) - b} holds by \eqref{equation - total sensitiviti of Litag bounded with beta}, which proves the desired upper bound of the total sensitivity for $L$.

\paragraph{Running time.}The overall time it takes to compute the distance from $\hat{P}$ to each line $\ell$ in $L$ is \\$T=O(d \cdot |L|\cdot |\hat{P}|)$. Since $|L| = n$ and by the size of $\abs{\hat{P}} = \abs{B} = O\of{\log n \of{dk \log k+\log(1/\delta)}}$ (see \eqref{equation - size of ab approx}), we get 
$$
T= O\of{d^2 n\log (n) k \log k }
.$$ 

By Theorem \ref{theorem - total sensitivity of k weighted points theorem}, the time it takes to compute the sensitivities $S_{L',k}(\ell')$ of each $\ell' \in L'$ is $ndk^{O(k)}$. Hence, the total running time is
$$
T + ndk^{O(k)}
=O\of{d^2 n\log (n) k \log k } + ndk^{O(k)}
,$$
which proves the lemma.
\end{proof}

\section{Algorithm~\ref{algorithm - coreset}: Computing Coresets from Sensitivities.}
In this thesis, the input is usually a set of lines in $\REAL^d$, but for the streaming case in the next section we compute coreset for union of (weighted) coresets and thus the input lines are also weighted. To set our result for a finite number of input lines (offline case), we use the folowing definitions and theorem from Feldman, Xuan and Kfir \cite{zahimsc}.

\begin{definition} [$\eps$-coreset \cite{zahimsc}]
Let $(P',Y,f,\loss)$ be a query space as in Definition~\ref{definition - k query space}. For an approximation error $\eps > 0$, the pair $S' = (S,u)$ is called an \emph{$\eps$-coreset} for the query space $(P',Y,f,\loss)$, if $S \subseteqÂ„ P, u : S \to [0, \infty)$, and for every $y \in Y$ we have
$$
(1 - \eps)f_{\loss}(P', y) \leq f_{\loss}(S',y) \leq (1 + \eps)f_{\loss}(P', y).
$$
\end{definition}

The following theorem proves that a coreset can be computed by sampling according to sensitivity of lines. The size of the coreset depends on the total sensitivity and the complexity (VC-dimension) of the query space, as well as the desired error $\eps$ and probability $\delta$ of failure.

\begin{theorem} [coreset construction \cite{zahimsc}] \label{theorem - the output of coreset is coreset}
Let
\begin{itemize}
\item $P' = (P,w)$ be a weighted set, $Y$ be a set and $f:P' \times Y \to [0,1)$.

\item $F^*_{P'} = \br{f_1, \ldots, f_n}$, where $f_i(y)= f(p_i, Q)$ for every $i \in [n]$ and $y \in Y$.

\item $(F^*_{P'}, Y, f, \norm{\cdot}_1)$ be a query space and $n=\abs{P}$.

\item $d^*$ be the dimension of the range space that induced by $Y$ and $F^*_{P'}$.

\item $s^* : P \to [0, \infty)$ s.t. $s^*(p)$ is the sensitivity of $p \in P$, and $s : P \to [0, \infty)$ be the sensitivity bound of $s^*$; See definition~\ref{definition - sensitivity}, just now it is the general case of any weighted set instead of lines. 

\item $t = \sum_{p \in P} s(p)$.

\item $\eps, \delta \in (0,1)$.

\item $c > 0$ be a universal constant that can be determined from the proof.

\item 
$$
m \geq \frac{c(t+1)}{\eps^2} \of{d' \log(t+1) + \log\of{\frac{1}{\delta}}}, 
$$
and

\item $(C,u)$ be the output weighted set of a call to $\textsc{Coreset-Framework}(P, w, s, m)$ (Algorithm 4 in \cite{zahimsc}).
\end{itemize}
Then the following holds
\begin{itemize}
\item With probability at least $1- \delta$, $C$ is an $\eps$-coreset of $(F^*_{P'}, Y, f, \norm{\cdot}_1)$.

\item $\abs{C} = m$.

\item $(C,u)$ can be computed in $O(n)$ time given $(P, w, s, m)$.

\item $u(P) \in [w(p), \sum_{q \in P}w(q)/m]$ for every $p \in C$.

\item $\sum_{p \in P} w(p) = \sum_{p \in C} u(p)$.
\end{itemize}
\end{theorem}

The following theorem is the main result of this thesis and its states that we can compute a small $\eps$-coreset for a set $L$ of lines and an $\eps \in (0,1)$.

\begin{theorem} [coreset for $k$-line median] \label{theorem - coreset offline}
Let $L'=(L,w)$ be a weighted set of $n$ lines in $\REAL^d$, $k \geq 1$ be an integer, $\eps,\delta\in(0,1)$ and $m >1$ be an integer such that
$$
m \geq \frac{cd^2k\log^2(k) \log^2(n)+\log(1/\delta)}{\eps^2},
$$ 
for some universal constant $c>0$ that can be determined from the proof, and $\Q_k=\br{B\subseteq\REAL^d\mid\abs{B}=k}$. Let $(S,u)$ be the output of a call to $\coreset(L,k,m)$; see Algorithm \ref{algorithm - coreset}. Then, with probability at least $1-\delta$, $(S,u)$ is an $\eps$-coreset for the query space $(F^*_{L'} ,\Q_k, \dist, \norm{\cdot}_1)$, where $F^*_{L'}$ is defined as in Corollary~\ref{corollary - the range space is O(dklog k)}. Moreover, $(S,u)$ can be computed in time
$$
O\of{d^2 n\log (n) k \log k} + ndk^{O(k)}
.$$
\end{theorem}

\begin{proof}
Let $b >0$ be an integer. Let $B$ be the bi-criteria approximation set that is computed in Line 2 of Alg.~\ref{algorithm - coreset}. Then substituting  $L, k$ and $\delta$ in Theorem~\ref{theorem - the output of a,b approx is indeed a,b approx} yields taht with probability $1-\delta$, $B$ is an $\of{\alpha,\beta}$-approximation of $L$, where $\alpha= 4 \rho^2$ and $\beta= O\of{d\log (n) \log k}$.

By Lemma~\ref{theorem - total sensitivity can be computed given a,b-approx}, given $B$ we can compute a sensitivity bound $S_{L,k}(\ell) \geq  S^*_{L,k}(\ell)$ for each line $\ell \in L$, s.t. the total sensitivity is then bounded by $S_k(L) \in \alpha \beta k^{O(k)}\log n$. Furthermore, by Corollary~\ref{corollary - the range space is O(dklog k)}, the VC-dimension of the range space that is induced by $F^*_{L'}$ and $\Q_k$ is $d^* \in O(dk \log k)$.

Substituting the set of input lines $P = L$, the VC-dimension $d^*$ and the total sensitivity $t = S_k(L) \in \alpha \beta k^{O(k)}\log n$ in Theorem~\ref{theorem - the output of coreset is coreset}, yields that $(S,u)$ is an $\eps$-coreset for $(F^*_L ,\Q_k, \dist, \norm{\cdot}_1)$, with probability at least $1-\delta$.

\paragraph{Running time.} By Theorem \ref{theorem - the output of a,b approx is indeed a,b approx}, we can compute the \abapprox~of $L$ in\\$O\of{d^2nk\log k +j \log n}$ time, for $j \in O(dk \log k)$, as defined in Line 1 of Alg.~\ref{algorithm - coreset}. By Lemma \ref{theorem - total sensitivity can be computed given a,b-approx}, the time it takes to compute the sensitivity for every line $\ell \in L$ is 
$$
O\of{d^2 n\log (n) k \log k } + ndk^{O(k)}
.$$ 
Since the rest of the algorithm takes time that is linear in the input size, we get a total running time of
$$
O\of{d^2 n\log (n) k \log k} + ndk^{O(k)}
,$$
and that proves the theorem.

\end{proof}

\section{Coreset for Streaming Data} \label{section - Coreset for Streaming Data}

In the previous section we showed how to compute an $\eps$-coreset for a query space$(F^*_L,\Q_k,\dist,\norm{\cdot}_1)$. However, we assumed that the input is finite and stored in memory. In this section we generalize the result to support streaming data. In this model, we are given a (possibly infinite) stream of lines. Our goal is to maintain the $\eps$-coreset $C_n$ for the (first) $n$ lines in the stream that we saw till now (i.e., for every $n\geq 1$). The required memory to maintain the coreset, as the insertion time per point should be only poly-logarithmic in $n$.

In the following definition ``sequence'', is an ordered multi-set.

\begin{definition} [input stream \cite{zahimsc}]
Let $L = \br{\ell_1, \ell_2, \ldots}$ be a (possibly infinite, unweighted) ordered set of lines in $\REAL^d$. A \emph{stream of lines} from $L$ is a procedure whose $i$th call returns the $i$th line $\ell_i$ in a sequence of lines that are contained in $L$, for every $i \geq 1$.
\end{definition}

The idea behind the merge-and-reduce tree that is shown in Algorithm 4 in \cite{zahimsc} is to merge every pair of small subsets and then reduce them by half. The relevant question is what is the smallest size of input that our coreset construction can always reduce by at least half. The log-Lipschitz property below is needed for approximating the cumulative error during the construction of the tree.

\begin{definition} [halving function \cite{zahimsc}]
Let $\eps, \delta, r >0$. A non-decreasing function $\hlf : [0, \infty) \to [0, \infty)$ is an $(\eps, \delta, r)$-\emph{halving} function of a function $\size : [0,\infty)^4 \to [0,\infty)$ if for every $h \geq 1$ and $n =\hlf(h)$ we have 
$$
\size\of{2n, 2^hn, \frac{\eps}{h}, \frac{\delta}{h}} \leq n.
$$
and $\hlf$ is $r$-log-Lipschitz, i.e., for every $\Delta \geq 1$ we have $\hlf(\Delta h) \leq \Delta^r\hlf(h)$.
\end{definition}

The following corollary enables us to determine what is the smallest subset that can be cut by a half using the output size of our off-line coreset.

\begin{corollary} [\cite{zahimsc}] \label{corollary - hlf is halving function of size}
Let $\eps, \delta \in (0,1)$, and $u : (0, \infty) \to (e, \infty)$ such that $u(\cdot)$ is $r$-log-Lipschitz function for some $r \geq 1$. Let $b \geq 1$ and $\size : [0, \infty)^4 \to [0, \infty)$ be a function such that
$$
\size\of{2n,w,\frac{\eps}{h},\frac{\delta}{4^h}} \leq \of{\frac{u(h) \ln(w)}{hb}}^b,
$$
for every $h,n,w \geq e$. Let $\hlf : [0, \infty) \to [0, \infty)$ be a function such that
$$
\hlf(h) \geq \of{4u(h) \ln \of{4u(h)}}^b
$$
for every $h \geq 1$.Then, $\hlf$ is an $(\eps, \delta, 2br)$-halving function of the function $\size$.
\end{corollary}

\begin{definition} [$(\eps,\delta)$-coreset scheme \cite{zahimsc}]

Let $(P, Y, f, \loss)$ be a query space such that $P$ is an (unweighted, possibly infinite) set. Let 
$\size : [0, \infty)^4 \to [1, \infty)$ and $\timee : [0, \infty)^4 \to [0, \infty)$. Let 
$\bcoreset$ be an algorithm that gets as input a weighted set $Q' =(Q,w)$ such that $Q \subseteq P$, an approximation error $\eps > 0$ and a probability of failure $\delta \in (0, 1)$. The tuple $(\bcoreset, \size, \timee)$ is called an $(\eps,\delta)$-\emph{coreset scheme} for $(P, Y, f, \loss)$ if $(i)$-$(iii)$ hold as follows:\\
\\$(i)$ A call to $\bcoreset(Q', \eps, \delta)$  returns a weighted set $(S, u)$.\\
\\$(ii)$ With probability at least $1-\delta$, $(S, u)$ is an $\eps$-coreset of $(P, Y, f, \loss)$.\\
\\$(iii)$ The computation time of $(S,u)$ is $\timee\of{\abs{Q},\sum_{q \in Q}w(q)/\min_{p}w(p),\eps,\delta}$ and its size is
$$
\abs{S} \leq \size \of{\abs{Q}, \sum_{q \in Q}w(q)/\min_{p}w(p),\eps,\delta}
,$$
where the minimum is over every $p \in Q$ with a positive weight $w(p)$.
\end{definition}

\begin{definition} [$(\eps,\delta,r)$-mergable coreset scheme \cite{zahimsc}]
Let $(\coreset,\size,\timee)$ be an $(\eps,\delta)$-coreset scheme for the
query space $(P, Y, f, \loss)$, such that the total weight of the coreset and the input is the same, i.e. a call to $\coreset((Q,w),\eps,\delta)$ returns a weighted set $(S,u)$ whose overall weight is $\sum_{p \in S}u(p) = \sum_{p \in Q}w(p)$. 
Let $\hlf$ be an $(\eps,\delta,r)$-halving function for $\size$.
Then the tuple $(\coreset, \hlf, \timee, \size)$ is an $(\eps,\delta,r)$-\emph{mergable coreset scheme} for $(P,Y, f, \loss)$.
\end{definition}

The following theorem states a reduction from off-line coreset construction to a coreset that is maintained during streaming. The required memory and update time depends only logarithmically in the number $n$ of lines seen so far in the stream. It also depends on the halving function that corresponds to the coreset via $\hlf(\cdot)$. 

The theorem below holds for a specific $n$ with probability at least $1-\delta$. However, by the union bound we can replace $\delta$ by, say, $\delta/n^2$ and obtain, with high probability, a coreset $S'_n$ for each of the $n$ lines insertions, simultaneously.

\begin{theorem} [\cite{zahimsc}] \label{theorem - the output ofstream alg is coreset}
Let $(\coreset, \hlf, \timee, \size)$ be an $(\eps,\delta,r)$-mergable coreset scheme for $(P,Y,f,\loss)$, where $\hlf$ is an $(\eps,\delta,r)$-halving function of $\size$, $r\geq 1$ is a constant and $\eps,\delta\in(0,1/2)$. Let $stream$ be a stream of items from $P$. Let $S'_n$ be the $n$th output weighted set of a call to \textsc{Streaming-Coreset}$(stream,\eps/6,\delta/6,\coreset,\hlf)$; See Algorithm \ref{algorithm - streaming coreset}. Then, with a probability at least $1-\delta$,
\begin{itemize}
\item (Correctness) $S'_n$ is an $\eps$-coreset of $(P_n,Y,f,\loss)$, where $P_n$ is the first $n$ items in $stream$.

\item (Size) $\abs{S_n} \in \size \of{\hlf(h), n, \eps, \delta}$ for some constant $h \geq 1$.

\item (Memory) there are at most $b = \hlf(h) \cdot O(\log^{r+1}n)$ items in memory during the computation of $S'_n$.

\item (Update time) $S'_n$ is outputted in additional $t = O(\log n)\cdot \timee \of{b,n,\frac{\eps}{O(\log n)},\frac{\delta}{n^{O(1)}}}$ time after $S'_{n-1}$. 

\item (Overall time) $S'_n$ is computed in $nt$ time.

\end{itemize}

\end{theorem}

\begin{algorithm}[H]
	\caption{$\streamingcoreset(stream,\eps,\delta,\coreset,\hlf)$\label{algorithm - streaming coreset}}
	\begin{tabbing}
		\textbf{Input:} \quad\quad\= An input $stream$ of items from a set $P$, an error parameter $\eps\in(0,1/2)$, \\\quad\quad\>probability of success $\delta \in (0,1)$, an algorithm \coreset~and a function \\\quad\quad\>$\hlf:[0,\infty)\to[0,\infty)$. \\
		\textbf{Output:} \>A sequence $S'_1, S'_2,\ldots$ of coresets that satisfies Theorem \ref{theorem - the output ofstream alg is coreset}.

	\end{tabbing}
	\vspace{-0.3cm}
	
	\nl \For{every integer $h$ from $1$ to $\infty$}{
		
		\nl Set $S_i \coloneqq \emptyset$ for every integer $i \geq 0$ 
		
		\nl $T_{h-1} \coloneqq S_{h-1}$
		
		\nl \For{$\lceil 2^{h-1}\cdot \hlf(h) \rceil$ iterations}{
				
				\nl Read the next item $p$ in $stream$ and add it to $S_0$
				
				\nl \uIf{$\abs{S_0} = \hlf(h)$}{
    				
    				\nl $i \coloneqq 0 ; S \coloneqq \emptyset$
    				
    				\nl \While{ $ S \ne \emptyset$ } {
    					
    					\nl $S \coloneqq \coreset \of{S \cup S_i,\eps/h,\delta/4^h}$
    					
    					\nl $S_i \coloneqq \emptyset$
    					
    					\nl $i \coloneqq i +1$
    				}
				
				\nl $S_i \coloneqq S$    			
    			}	
    				\nl $S'_n \coloneqq \coreset \of{\of{\bigcup_{i=0}^{h-1}T_i} \cup \of{\bigcup_{i=0}^hS_i},\eps,\delta}$
				  
  }

}
	
	\nl \Return $S'_n$
\end{algorithm}

\begin{theorem} \label{theorem - summarize of streaming version for lines}
Let $stream=\br{\ell_1,\ell_2,\ldots}$ be a stream of lines in $\REAL^d$, and let $n>0$ denote the number of lines seen so far in the stream. Let $k \geq 1$ be an integer, $c>0$ be a suffice large constant, $\eps, \delta \in (0,1)$, and let 
$$
m \geq \frac{cd^2k\log^2(k)\log^2(n)\log(e/\delta)}{\eps^2}.
$$ 
For every $h \geq 1$ we define
\begin{equation} \label{equation - hlf(h) value}
\hlf(h) = \frac{ch^9m^3}{\ln^3(n)}.
\end{equation}
Let $S'_1,S'_2,\ldots$ be the output of a call to $\textsc{Streaming-Coreset}(stream, \eps/6, \delta/6,\coreset,\hlf)$, where a call to $\coreset((Q,w),\eps,\delta)$ returns a weighted set $S'=(S,u)$ whose overall weight is $\sum_{p \in S}u(p) = \sum_{p \in Q}w(p)$; See Alg. \ref{algorithm - streaming coreset}. Then, with probability at least $1-\delta$, the following hold. For every $n \geq 1$:\\
\\$(i)$ (Correctness) $S'_n$ is an $\eps$-coreset of $(L_n ,\Q_k, \dist, \norm{\cdot}_1)$, where $L_n$ is the first $n$ lines in $stream$.\\
\\$(ii)$ (Size)
$$
\abs{S'_n} \in O\of{\frac{m^3}{\ln^3n}}
.$$
$(iii)$ (Memory) there are
$$
b \in O\of{m^3}
$$
lines in memory during the streaming.\\
\\$(iv)$ (Update time) $S'_n$ is outputted in additional 
$$
\displaystyle t \in O\of{d^2 b\log (b) k \log k} + bdk^{O(k)}
$$
time after $S'_{n-1}$.\\
\\$(v)$ (Overall time) $S'_n$ is computed in $nt$ time.

\end{theorem}
\begin{proof}
Substituting $L = L_n$ in Theorem \ref{theorem - coreset offline} yields that\\$(\coreset, \size, \timee)$ is an $(\eps,\delta)$-coreset scheme for the query space $(F^*_{L_n}, \Q_k, \dist, \norm{\cdot}_1)$, where
\begin{equation} \label{equation - bound on coreset size}
\size(n,n,\eps,\delta) \leq \frac{cd^2k\log^2(k)\log^2(n)\log(e/\delta)}{\eps^2} \leq m,
\end{equation}
and
$$
\timee(n,n,\eps, \delta) \in O\of{d^2 n\log (n) k \log k} + ndk^{O(k)}
$$
Let $h, w \geq 1$. We have,
\begin{equation} \label{equation - first bound on size(2n,2n..)}
\size(2n,2n,\eps / h, \delta / 4^h)
\leq \frac{c h^2 d k^k \log^2(n)\log \frac{4h}{\delta}}{\eps^2}
\leq \ln (4) h^3 m ,
\end{equation}
where the first inequality is from the substitution of $2n$ instead of $n$ in \eqref{equation - bound on coreset size}, and the second inequality is by the definition of $m$. 
Let
\begin{equation} \label{equation - u(h) value}
u(h) = \frac{2h \of{6 h^3 m}^{\frac{1}{2}}}{\ln n}
.
\end{equation}
Then, we can bound $\size(2n,2n,\eps / h, \delta / 4^h)$ by $\hlf(h)$ as follows:

\begin{align} 
\size(2n,2n,\eps / h, \delta / 4^h)
&  \leq 6 h^3 m \label{equation - bound on size(2n,2n..) - a}\\
& \leq \of{\frac{u(h) \ln n}{2h}}^2 \label{equation - bound on size(2n,2n..) - b}\\
& \leq  10 u^3(h) \nonumber \\
& \leq \hlf(h), \label{equation - bound on size(2n,2n..) - c}
\end{align}

where \eqref{equation - bound on size(2n,2n..) - a} is by \eqref{equation - first bound on size(2n,2n..)}, \eqref{equation - bound on size(2n,2n..) - b} is by the definition of $u(h)$ in \eqref{equation - u(h) value}, and \eqref{equation - bound on size(2n,2n..) - c} is obtained from the definition of $\hlf(h)$ in \eqref{equation - hlf(h) value}.

Substituting $r=3/2$ and $b=2$ in Corollary~\ref{corollary - hlf is halving function of size} yields that $\hlf$ is an $(\eps,\delta,6)$-halving function of $\size$. Substituting
$$
\hlf(1) \in O\of{\frac{m^3}{\ln^3 n}}
,$$
and
$$
\timee\of{b,n,\frac{\eps}{O(\log n)},\frac{\delta}{n^{O(1)}}} \in 
O\of{d^2 b\log (b) k \log k} + bdk^{O(k)}
$$

in Theorem \ref{theorem - the output ofstream alg is coreset} then proves Theorem \ref{theorem - summarize of streaming version for lines} for the query space $(F^*_{L_n}, \Q_k, \dist, \norm{\cdot}_1)$.
\end{proof}

\chapter{Experimental Results} \label{section - experimental results}

Following motivation to narrow the gap between the theoretical and practical fields, experiments took a dominant place during research, and mainly divided into two parts: (1) The main problem we solve in this thesis - the $k$-line median experiments, where a set of $k$ points (centroids) was computed repeatedly given different distributed sets of lines in $\REAL^2$, $\REAL^3$ and higher dimensions.\\(2) Anomalies detection, which is one of the most fundumental problems in machine learning. In this part anomalies were calculated for variety of input sets of points in $\REAL^d$, using a method from Schulman and Feldman \cite{analysincompletedata} for outliers detection. The need for (2) is a reduction we made to \cite{analysincompletedata} during sensitivity calculation for our coreset for $k$-line median; See Algorithm~\ref{algorithm - sensitivityoftranslatedlines}.

\paragraph{Software.}
We implemented our coreset construction from Algorithm~\ref{algorithm - coreset} and its sub-procedures in Python V. 3.6. We make use of the MKL package \cite{mkl} to improve its performance, but it is not necessary in order to run it. 

\paragraph{Data Sets.}
We evaluate our system on two types of data sets: synthetic data generated with carefully controlled parameters, and real data of roads map from the "Open Street Map" Dataset \cite{openstreetmap} and "SimpleHome XCS7 1002 WHT Security Camera" from the the "UCI Machine Learning Repository" Dataset \cite{uci}.

The roads dataset \cite{openstreetmap} contains $n=10,000$ roads in China from the "Open Street Map" dataset (Fig.~\ref{figure - experimental results} plot (a)), where each road is represented as a 2-dimensional segment that was stretched into an infinite line on a plane. Synthetic data of $n=10,000$ lines was generated as well (Fig.~\ref{figure - experimental results} plot (b)).

The security cameras dataset \cite{uci} contains 40,000 data observations coming from security cameras raw data, each contains 114 features $(d=114)$.  
\section{$k$-line median Experiments}

\subsection{The Experiment}
\paragraph{Experiments on offline data.}
At each iteration of the experiment, a sample whose size increases in every iteration was taken by coreset and by the competitor Random Sample Consensus (RANSAC) \cite{ransac}. The $k$-line median of each sample was calculated by the standard EM algorithm and by our constant approximation $k$-line median algorithm (which is an exhaustive search for $k$-line median over the output of \constantapprox; See Alg.~\ref{algorithm - 4 approx}), and the error was measured by the sum of squared distances from the original set of lines to the $k$ medians that were calculated on the sample.

\paragraph{Experiments on streaming data.}
To produce the main streaming experiment, we created 6 clusters, each consist of 2, 3, 4, 5, 7 and 10 machines on Amazon EC2 platform \cite{amazoneec2}, each cluster computed coreset for different sets of lines from \cite{uci}, using the  merge-and-reduce tree technique. The idea behind the merge-and-reduce tree that is shown in Algorithm 4 in \cite{zahimsc} is to merge every pair of small subsets and then reduce them by half, and since a union of coreset is a global core-set for the union of original data, then distributed calculation in a cluster is a natural approach. We show that the coreset construction running time decreases linearly as the number of machines in the machines cluster increases, where coreset construction was measured 3 different times on 3 different number of centers.

\subsection{Results}
\paragraph{Experiments on offline data.}
In Plot (a) in Graph~\ref{figure - experimental results}, when the sample size was $m=700$ lines out of 10,000 given lines, the coreset error and variance were 1.86 and 0.16, respectively, that is an error of $\eps=0.86$, for a sample size of $m=\lceil 602/\eps \rceil$ lines. On the other hand, the  error and variance of the competitor algorithm with the same sample size were 2.62 and 0.26. This implies that our coreset is more accurate and stable than RANSAC, and that our mathematically provable constant approximation algorithm for $k$-line median works better than a standard EM algorithm also in practice.

\paragraph{Experiments on streaming data.}
Plot (c) in Graph~\ref{figure - experimental results} demonstrates the size of the merge-and-reduce streaming coreset tree during the streaming, which is logarithmic in the number of lines we streamed so far. In Plot (d) in Graph~\ref{figure - experimental results}, we can see how the coreset construction running time decreases linearly as the number of machines in the machines cluster increases (parameters are written in the chart's title), where coreset construction was measured 3 different times on 3 different number of centers. Note that the decrease rate is almost linear in the cluster's machines number and not exactly, due to overhead of communications and I/O.
%figure - experimental results
\begin{figure} 
\includegraphics[scale=0.3]{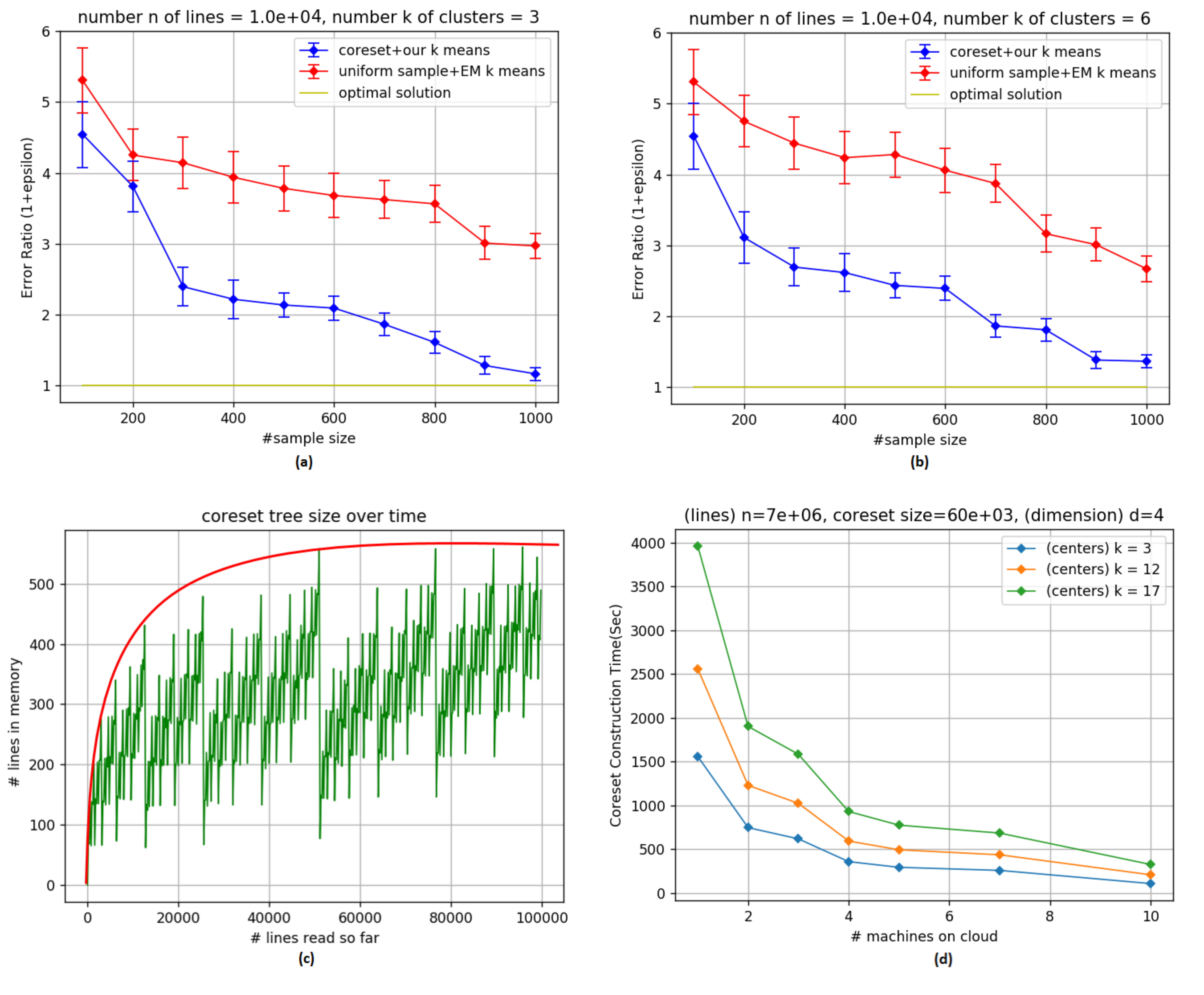}
		\caption{\textbf{$k$-line median experimental results.} Graphs (a) and (b) show that larger sample results in a smaller error by coreset and uniform sampling. Graph (c) shows  that the required memory increases only logarithmically with the number of lines that were read so far from the stream. Graph (d) illustrates  how the coreset construction time decreases near-linearly in the number of machines in Amazon EC2 cluster.}
		\centering		
		\label{figure - experimental results}
\end{figure}

\clearpage

\section{Anomalies Detection Experiments}

One of the steps during the sensitivity calculation of each line $\ell \in L$ for a query space of all the set of $k$ points in $\REAL^d$, is to compute the sensitivities of $n$ points in relate to a query space of all the families of $2k$ points in space (Line 6 in Alg. \ref{algorithm - sensitivityoftranslatedlines}). To this end, we implemented the suggested sensitivity by Langebreg and Schulman~\cite{analysincompletedata}, and measured its performance as well on different outliers detection tasks, in order to examine its correctness, quality and other parameters; See Fig. \ref{figure - weighted_centers_results}.

\subsection{The Experiment}
As in the $k$-line median experiments, the first experiment was a comparison between coreset and RANSAC \cite{ransac}, where in each iteration of the experiment, a sample whose size increases in every iteration was taken by coreset and by the competitor Random Sample Consensus - RANSAC \cite{ransac}. Outliers were computed for each sample by a standard EM algorithm - that is running the following procedure iteratively until a tunable stop condition is being achieved: (1) chose and remove randomly $m$ points (later to be the outliers) out of the input set. (2) partition the remain points in the set into $k$ clusters and (3) sum the distances from each point to its nearest center in the cluster.

The second experiment is being applied as the first one above, but now the error is being measured as the number of outliers each sample (coreset and RANSAC) yields on average on 1000 iterations.

In the third experiment we visualized the quality of the outliers that we found with coreset and with RANSAC using the EM algorithm.

\subsection{Results}

In Plot (a) in Graph~\ref{figure - weighted_centers_results}, we can see in one of the measurements that we get an error of $\varepsilon = 0.4$ for a sample size of $m=1000=\frac{400}{\varepsilon}$ points out of $40000$ given data points, while the  error we got from the competitor algorithm was $1.8$. 

In Plot (b) in Graph~\ref{figure - weighted_centers_results} we show, for example, that running the EM-algorithm for outliers detection on coreset with size of $m=4000$ - which is $10\%$ of the entire data - yields in average a detection of $3.2$ outliers out of 6 , while running it on the same amount of sample we got from the competitor algorithm detects less than 1.

Plot (c) in Graph~\ref{figure - weighted_centers_results} illustrates the quality of the outliers we get from coreset and from RANSAC - it mainly point that coreset contains the outliers while RANSAC does not.

%figure - weighted centers results
\begin{figure}
		\centering
		\includegraphics[width=18cm]{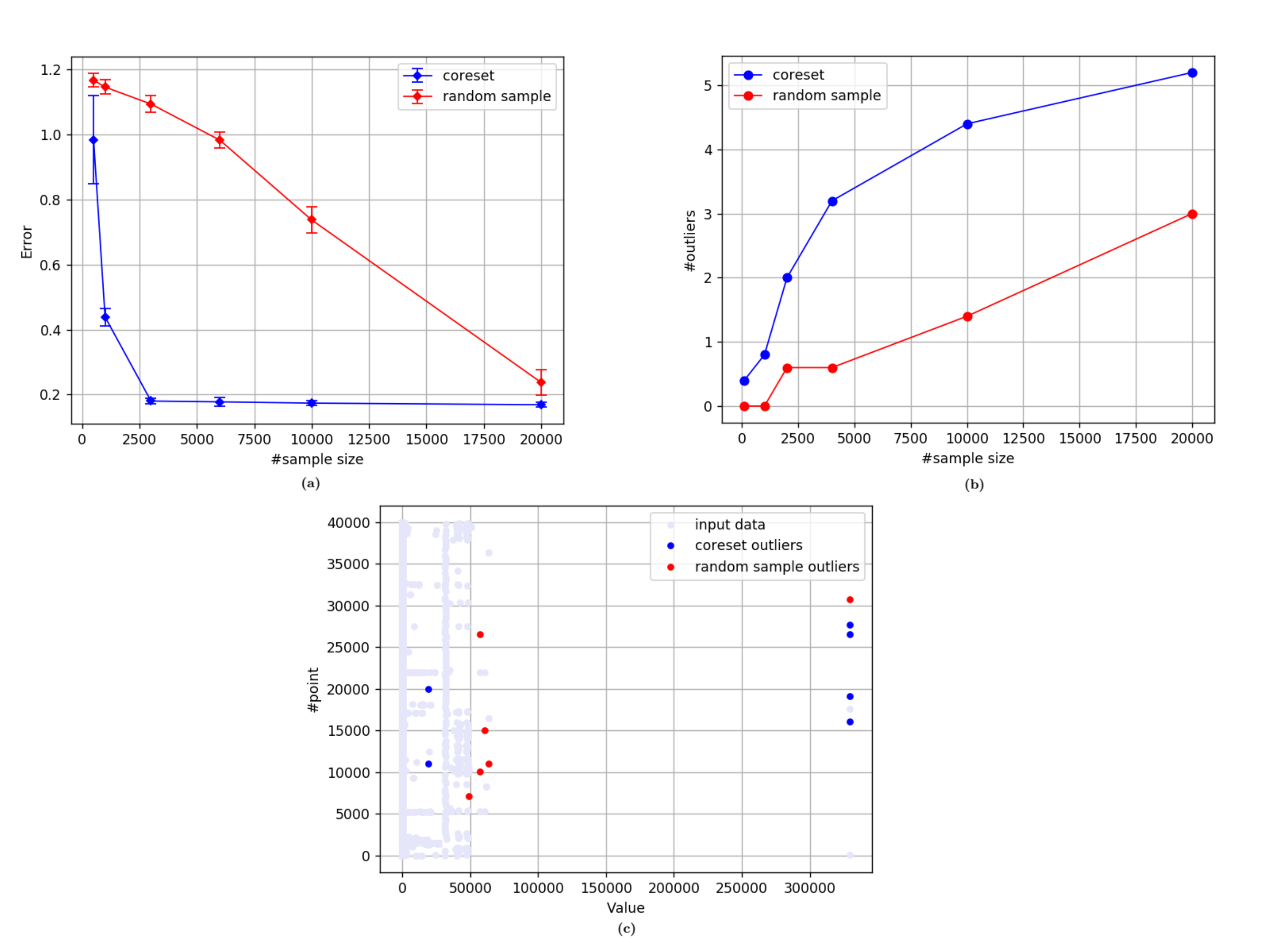}
\caption{\textbf{Anomaly detection experimental results.} Graphs (a) and (b) show how the error decreases as the size of the sample increases. Graph (c) shows the sampling quality - while coreset outliers are dominant outliers - globally in relate to the input set and localy inside the clusters, while RANSAC outliers are much less dominant.
}	
		\label{figure - weighted_centers_results}
\end{figure}

\clearpage
\section{Conclusions and Future Work}
In this thesis we purposed an algorithm that computes an $\eps$-coreset of size near-logarithmic in the input. Moreover, we suggest a streaming algorithm that computes a $(1+\eps)$-approximation for the $k$-line median of any set of lines that is distributed among $M$ machines, where each machine needs to send only near-logarithmic number of input lines to the main server for its computations.
Other future work will consider an input of $j$-dimensional affine sub-spaces in $\REAL^d$ (here input of lines is a private case of $j=1$), where the motivation is multiple missing entries completion.

\include{biblio}
\end{document}

%% file: cover.tex
\title{$k$-Clustering of Lines and Its Applications}

\author{Yair Marom}

\department{Department of Computer Science}

\degree{Master of Science in Computer Science}

\degreemonth{September,}
\degreeyear{2019}
\thesisdate{September 25, 2019}

\supervisor{Dan Feldman}{Ph.D.}

\chairman{Ronen Shaltiel}{Chairman, Department Committee on Graduate Theses}

\maketitle

% The abstractpage environment sets up everything on the page except
% the text itself.  The title and other header material are put at the
% top of the page, and the supervisors are listed at the bottom.  A
% new page is begun both before and after.  Of course, an abstract may
% be more than one page itself.  If you need more control over the
% format of the page, you can use the abstract environment, which puts
% the word "Abstract" at the beginning and single spaces its text.

%% You can either \input (*not* \include) your abstract file, or you can put
%% the text of the abstract directly between the \begin{abstractpage} and
%% \end{abstractpage} commands.

% First copy: start a new page, and save the page number.
%\cleardoublepage
% Uncomment the next line if you do NOT want a page number on your
% abstract and acknowledgments pages.
% \pagestyle{empty}
%\setcounter{savepage}{\thepage}
%\begin{abstractpage}
%\input{abstract}
%\end{abstractpage}

% Additional copy: start a new page, and reset the page number.  This way,
% the second copy of the abstract is not counted as separate pages.
% Uncomment the next 6 lines if you need two copies of the abstract
% page.
% \setcounter{page}{\thesavepage}
% \begin{abstractpage}
% \input{abstract}
% \end{abstractpage}

\cleardoublepage

\section*{Acknowledgments}

The novel results of this work were possible due to the help and support of several people, I would like to take the opportunity to show my appreciation.\\

\noindent First and foremost, I would like to thank my supervisor, Dr. Dan Feldman. The door to Dr.Feldman's office was always open whenever I had a question about my research or writing. He consistently allowed this paper to be my own work, but steered me in the right direction whenever he thought I needed it. It was a great honor for me to work with and learn from such a great researcher.\\

\noindent Finally, I must express my very profound gratitude to my family and friends for providing me with unfailing support and continuous encouragement throughout my years of study and through the process of researching and writing this thesis. This accomplishment would not have been possible without them. Thank you.

%%%%%%%%%%%%%%%%%%%%%%%%%%%%%%%%%%%%%%%%%%%%%%%%%%%%%%%%%%%%%%%%%%%%%%
% -*-latex-*-

%% file: contents.tex
  % -*- Mode:TeX -*-
%% This file simply contains the commands that actually generate the table of
%% contents and lists of figures and tables.  You can omit any or all of
%% these files by simply taking out the appropriate command.  For more
%% information on these files, see appendix C.3.3 of the LaTeX manual. 
\renewcommand{\cftpartleader}{\cftdotfill{\cftdotsep}} % for parts
\renewcommand{\cftchapleader}{\cftdotfill{\cftdotsep}} % forchapters
\renewcommand{\cftsecleader}{\cftdotfill{\cftdotsep}}
\tableofcontents
%\newpage
%\listoffigures
%\newpage
%\listoftables

%% file: biblio.tex
%% This defines the bibliography file (main.bib) and the bibliography style.
%% If you want to create a bibliography file by hand, change the contents of
%% this file to a `thebibliography' environment.  For more information 
%% see section 4.3 of the LaTeX manual.
\begin{singlespace}
\addcontentsline{toc}{chapter}{Bibliography}
\bibliography{references}
\bibliographystyle{plain}
\end{singlespace}